\catcode`\@=11
\newif\if@fewtab\@fewtabtrue
{\count255=\time\divide\count255 by 60
\xdef\hourmin{\number\count255}
\multiply\count255 by-60\advance\count255 by\time
\xdef\hourmin{\hourmin:\ifnum\count255<10 0\fi\the\count255}}
\def\ps@draft{\let\@mkboth\@gobbletwo
    \def\@oddfoot{\hbox to 7 cm{\tiny \versionno
       \hfil}\hskip -7cm\hfil\rm\thepage \hfil {\tiny\draftdate}}
    \def\@oddhead{}
    \def\@evenhead{}\let\@evenfoot\@oddfoot}
\def\draftdate{\number\month/\number\day/\number\year\ \ \ \hourmin }

\def\citen#1{\if@filesw \immediate\write \@auxout {\string\citation{#1}}\fi%
\@tempcntb\m@ne \let\@h@ld\relax \def\@citea{}%
\@for \@citeb:=#1\do {\@ifundefined {b@\@citeb}%
    {\@h@ld\@citea\@tempcntb\m@ne{\bf ?}%
    \@warning {Citation `\@citeb ' on page \thepage \space undefined}}%
    {\@tempcnta\@tempcntb \advance\@tempcnta\@ne
    \setbox\z@\hbox\bgroup\ifcat0\csname b@\@citeb \endcsname \relax
    \egroup \@tempcntb\number\csname b@\@citeb \endcsname \relax
    \else \egroup \@tempcntb\m@ne \fi \ifnum\@tempcnta=\@tempcntb
    \ifx\@h@ld\relax \edef \@h@ld{\@citea\csname b@\@citeb\endcsname}%
    \else \edef\@h@ld{\hbox{--}\penalty\@highpenalty
    \csname b@\@citeb\endcsname}\fi
    \else \@h@ld\@citea\csname b@\@citeb \endcsname \let\@h@ld\relax \fi}%
\def\@citea{,\penalty\@highpenalty\hskip.13em plus.13em minus.13em}}\@h@ld}
\def\@citex[#1]#2{\@cite{\citen{#2}}{#1}}%
\def\@cite#1#2{\leavevmode\unskip\ifnum\lastpenalty=\z@\penalty\@highpenalty\fi%
  \ [{\multiply\@highpenalty 3 #1%
  \if@tempswa,\penalty\@highpenalty\ #2\fi}]}   %
\makeatother 
\catcode`\@=12

\def\aa            {{\ensuremath{\mathcal A_{\mathcal A}}}}
\def\be            {\begin{equation}}
\def\bearl         {\begin{array}{l}}
\def\bearll        {\begin{array}{ll}}
\def\Bimod         {\mbox{bimod}}
\def\BIMOD         {\mbox{$\mathscr B${\sl imod}}}
\def\boti          {\,{\boxtimes}\,}
\def\C             {{\ensuremath{\mathcal C}}}
\def\cala          {{\ensuremath{\mathcal A}}}
\def\calb          {{\ensuremath{\mathcal B}}}
\def\calc          {{\ensuremath{\mathcal C}}}
\def\cald          {{\ensuremath{\mathcal D}}}
\def\calm          {{\ensuremath{\mathcal M}}}
\def\calw          {{\ensuremath{\mathcal W}}}
\def\calx          {{\ensuremath{\mathcal X}}}
\def\cc            {{\ensuremath{\mathcal C_{\mathcal C}}}}

\def\Ce            {{\ensuremath{\mathcal C^\text e}}}
\def\cir           {\,{\circ}\,}
\def\complex       {{\ensuremath{\mathbb C}}}
\def\complexx      {{\ensuremath{\mathbb C}^\times}}

\def\D             {{\ensuremath{\mathcal D}}}

\def\dim           {\mathrm{dim}}

\def\dimk          {\mathrm{dim}_\ko}

\def\ee            {\end{equation}}

\def\eear          {\end{array}}
\def\End           {{\ensuremath{\mathrm{End}}}}
\def\END           {{\ensuremath{\mathcal E\!\mbox{\sl nd}}}}
\def\ENDA          {{\ensuremath{\mathcal E\!\mbox{\sl nd}_\cala}}}

\def\ENDC          {{\ensuremath{\mathcal E\!\mbox{\sl nd}_\C}}}

\def\eps           {\varepsilon}
\def\eq            {\,{=}\,}
\def\Fcala         {F_{\!\cala}}
\newcommand\Ffrom[1] {{\ensuremath{F_{\! #1\leftarrow}}}}
\def\fpdim         {\mathrm{FPdim}}
\newcommand\Fto[1] {{\ensuremath{F_{\!\to #1}}}}

\def\Hom           {{\ensuremath{\mathrm{Hom}}}}
\def\HOM           {{\ensuremath{\mathcal H\!\mbox{\sl om}}}}

\def\HOMC          {{\ensuremath{\mathcal H\!\mbox{\sl om}_\C}}}
\def\Homk          {{\ensuremath{\mathrm{Hom}_\ko}}}
\newcommand\hsp[1] {\mbox{\hspace{#1 em}}}
\def\id            {\mbox{\sl id}}
\def\Id            {\mbox{\sl Id}}

\def\idsm          {\mbox{\footnotesize\sl id}}

\def\ii            {{\rm i}}
\def\Im            {\mathrm{Im}}
\def\iN            {\,{\in}\,}
\def\Ind           {{\ensuremath{\mathrm{Ind}}}}
\def\Irr           {{\mathrm{Irr}}}
\def\ko            {{\ensuremath{\Bbbk}}}
\newcommand\labl[1]{\label{#1}\ee}
\def\Mod           {\mbox{mod}}
\def\MOD           {\mbox{$\mathscr M${\sl od}}}

\def\nxt           {\raisebox{.08em}{\rule{.41em}{.41em}}~\,}
\def\one           {{\bf1}}

\def\oti           {\,{\otimes}\,}

\def\otia          {\,{\otimes_{\!\cala}}\,}
\def\otiae         {\,{\otimes_{\!\cala_1^{}}}\,}
\def\otiaz         {\,{\otimes_{\!\cala_2^{}}}\,}

\def\otie          {\,{\otimes_1}\,}
\def\otik          {\,{\otimes_\ko}\,}

\def\otiz          {\,{\otimes_2}\,}
\def\qquand        {\qquad{\rm and}\qquad}

\def\rev           {{\mathrm{rev}}}

\def\tc            {{T_{\mathcal C}}}
\newcommand\tFfrom[1] {{\ensuremath{\widetilde F_{\! #1\leftarrow}}}}
\newcommand\tFto[1] {{\ensuremath{\widetilde F_{\!\to #1}}}}
\newcommand\tFtofr[1] {{\ensuremath{\widetilde F_{\!\to #1\leftarrow}}}}

\def\threedim      {three-di\-men\-si\-o\-nal}
\def\Times         {\,{\times}\,}
\def\To            {\,{\to}\,}
\def\twodim        {two-di\-men\-si\-o\-nal}

\def\Vectc         {\ensuremath{\mathcal V\mbox{\sl ect}_\complex}}
\def\Vectk         {\ensuremath{\mathcal V\mbox{\sl ect}_\ko}}
 \newcommand\void[1]{}
\def\Wilsa         {{\ensuremath{\mathcal W_a}}}
\def\Z             {{\ensuremath{\mathcal Z}}}
\def\zet           {{\ensuremath{\mathbb Z}}}

\documentclass[12pt]{article}
\usepackage{latexsym, amsmath, amsthm, amsfonts} 
\usepackage{enumerate, amssymb, xspace, xypic}
\usepackage[all]{xy}
\usepackage[mathscr]{eucal}
\usepackage{graphicx} \usepackage{rotating}
\usepackage{epstopdf,hyperref}
\usepackage{color}

\newtheorem{thm}{Theorem}

\newtheorem{lem}[thm]{Lemma}
\newtheorem{lemma}[thm]{Lemma}
\newtheorem{prop}[thm]{Proposition}

\theoremstyle{definition}

\newtheorem{defi}[thm]{Definition}
\newtheorem{Definition}[thm]{Definition}
\newtheorem{rem}[thm]{Remark}

\begin{document}

\def\cir{\,{\circ}\,} 
\numberwithin{equation}{section}
\numberwithin{thm}{section}
                                     
\begin{flushright}
   {\sf ZMP-HH/12-5}\\
   {\sf Hamburger$\;$Beitr\"age$\;$zur$\;$Mathematik$\;$Nr.$\;$433}\\[2mm]
   March 2012
\end{flushright}
\vskip 3.5em
                                   
\begin{center}
\begin{tabular}c \Large\bf 
   Bicategories for boundary conditions \\[2mm] \Large\bf
   and for surface defects in 3-d TFT 
\end{tabular}\vskip 2.5em
                                     
  ~J\"urgen Fuchs\,$^{\,a}$,~
  ~Christoph Schweigert\,$^{\,b}$,~
  ~Alessandro Valentino\,$^{\,b}$

\vskip 9mm

  \it$^a$
  Teoretisk fysik, \ Karlstads Universitet\\
  Universitetsgatan 21, \ S\,--\,651\,88\, Karlstad \\[4pt]
  \it$^b$
  Fachbereich Mathematik, \ Universit\"at Hamburg\\
  Bereich Algebra und Zahlentheorie\\
  Bundesstra\ss e 55, \ D\,--\,20\,146\, Hamburg
                   
\end{center}
                     
\vskip 5.3em

\noindent{\sc Abstract}\\[3pt]
We analyze topological boundary conditions and topological surface 
defects in three-dimensional topological field theories of 
Reshetikhin-Turaev type based on arbitrary modular tensor categories. 
Boundary conditions are described by central functors that lift to 
trivializations in the Witt group of modular tensor categories. 
The bicategory of boundary conditions can be described through
the bicategory of module categories over any such trivialization. 
A similar description is obtained for topological surface defects. 
Using string diagrams for bicategories we also establish a precise 
relation between special symmetric Frobenius algebras and Wilson lines 
involving special defects. 
We compare our results with previous work of Kapustin-Saulina and
of Kitaev-Kong on boundary conditions and surface defects in abelian
Chern-Simons theories and in Turaev-Viro type TFTs, respectively.

\newpage

\section{Introduction}\label{intro}

An insight gained in recent years in the study of quantum field theories
is that interesting effects are captured when allowing for codimension-one
defects, i.e.\ interfaces between regions on which two different theories 
are living. Depending on the application, it is sensible to impose specific 
kinds of conditions on such interfaces; for instance, in integrable field 
theories, integrable defects, as considered e.g.\ in \cite{demS2,bocz},
are naturally of interest. In two-dimensional rational conformal field 
theories, the study of totally transmissive defect lines (see e.g.\ 
\cite{watt8,pezu5,bddo,ffrs5}) has produced structural
information about non-chiral symmetries and Kramers-Wannier-type dualities.
It has also become apparent that boundaries and defects are close relatives.

In this paper we concentrate on topological quantum field theories in three 
dimensions (TFTs), specifically on theories that include (compact) Chern-Simons 
theories. While for the latter subclass a Lagrangian formulation is
available, in the general case considered here we work within the combinatorial
approach of Reshetikhin and Turaev \cite{retu2} type, which associates a TFT
to any semisimple modular tensor category.
This includes TFTs of Turaev-Viro type as well.

In the \threedim\ situation the simplest codimension-one structures
are surfaces that constitute either a defect surface or a \twodim\ boundary. 
It is worth stressing that here the term boundary refers, as in  \cite{kaSau2},
to a brim at which `the \threedim\ world ends'. Such brim boundaries must
in particular not be confused with the 
cut-and-paste boundaries that commonly occur (see e.g.\ \cite{TUra}) in these 
theories. Boundaries of the latter kind arise when a three-manifold is cut
into more elementary three-manifolds with boundaries; accordingly, their 
function is to account for locality and to allow for sewing, or cut-and-paste, 
procedures. Both classes of boundaries are geometric boundaries of 
three-manifolds, but cut-and-paste boundaries come with additional local 
(chiral) degrees of freedom and can support vector spaces of conformal blocks. 
In contrast, brim boundaries need not involve any of those structures.
Note that the distinction between two different kinds of boundaries is not 
specific to three dimensions. In two dimensions, in the discussion of 
so-called open-closed theories, such a distinction is standard; see e.g.\ 
\cite[Sect.\ 3]{moor10}, where intervals corresponding to in- and out-going 
open strings are distinguished from ``free boundaries'' corresponding to 
the ends of an open string ``moving along a D-brane.''

Among the Reshetikhin-Turaev type theories there are in particular TFTs 
constructed from lattices, which have e.g.\ been prominent (see, for
instance, \cite{fcgkkmst} for a detailed discussion) in the discussion of 
universality classes of quantum Hall systems. Thus in this particular case, 
our results may have applications to topological interfaces 
with gapped excitations between two quantum Hall fluids.
Our discussion applies, however, to arbitrary semisimple modular tensor 
categories and does not rely on any specific aspects of lattice models.

There is no guarantee that for a given quantum field theory a consistent 
defect or boundary 
condition exists at all. In particular there can be theories that make 
perfect sense in the bulk, but cannot be consistently extended to the 
boundary. On the other hand, if consistent codimension-one defects,
or boundary conditions, do exist, they will typically not be unique. It 
is then natural to study interfaces between such lower-dimensional regions
as well, i.e.\ interfaces of codimension two. In our case of three-dimensional
topological field theories, these are generalized Wilson lines. (In other 
words, the brim boundaries we consider can contain such Wilson lines.
In contrast, this is not possible for cut-and-paste boundaries. On the
other hand, bulk Wilson lines can end on either kind of boundary -- in the
case of cut-and-paste boundaries, they end on marked points.)
Again, such generalized Wilson lines need not exist, but again, if they
do exist, then they need not be unique, so that the game can be repeated
one step further.

Hereby we arrive at a four-layered structure: At the top level, we associate
a topological field theory of Reshetikhin-Turaev
type to each three-dimensional part of a stratified three-dimensional manifold.
For \twodim\ parts we deal with physical boundaries or with \twodim\ defects,
for which we must choose a boundary condition, respectively, in the same spirit, 
an additional datum that describes the type, or `color' of the defect. Such a 
datum has been called a \emph{surface operator} in \cite{kaSau2}; we
prefer the term \emph{surface defect} instead. The third layer of structure
consists of one-dimensional structures labeled by generalized Wilson lines 
that separate boundaries or surface defects. And finally, generalized Wilson 
lines can fuse and split at point-like defects, which may be interpreted as 
local field insertions and constitute the fourth layer of structure.

\medskip

The basic questions we are addressing in this paper can thus be posed as follows:
  \def\leftmargini{1.7em}~\\[-1.55em]\begin{itemize}\addtolength{\itemsep}{-7pt}
  \item[(1)]
\emph{Given a three-dimensional region with non-empty boundary for which the 
TFT of Reshe\-tik\-hin-Turaev type in the interior is labeled by 
a modular tensor category \C, what are the data describing the 
types of \underline{\em topological boundary conditions} on the boundary}?
  \item[(2)]
\emph{Given two three-dimensional regions separated by a two-dimensional
interface, for which the TFTs of Reshetikhin-Turaev type
in the two regions are labeled by modular tensor categories 
$\C_1$ and $\C_2$, respectively, what are the data describing the 
types of \underline{\em topological surface defects} on the interface}?
\end{itemize}

\smallskip

The key in our analysis of these issues is the following process:
a Wilson line in the three-dimensional bulk can be moved ``adiabatically'' 
into the boundary or into a defect surface. This has already 
been studied in \cite[Sect.\,5.2]{kaSau2}, and a similar process in two 
dimensions has been considered in \cite[Sect.\,4.1]{dakr2}.
A careful analysis of this process 
allows us to give a complete answer to both questions, 
including in particular a criterion for the existence of non-trivial 
solutions. The analysis yields in particular a model-independent 
generalization of results that have been obtained in \cite{kaSau2}
for abelian Chern-Simons theories using a Lagrangian description.

Our considerations involve mathematical ingredients that, to the best 
of our knowledge, have not been applied to Reshetikhin-Turaev type
TFTs before. Many of them come from higher category theory, 
like aspects of fusion categories \cite{etno,enoM}
and of braided fusion categories \cite{dgno2}, and specifically 
the notions \cite{dmno} of central functors and of the Witt group of 
non-degenerate fusion categories. This group naturally generalizes the 
classical Witt group of lattices; it has been originally devised as a 
tool in the classification of modular tensor categories.  Since some 
familiarity with such concepts is required for appreciating our analysis,
we collect the pertinent mathematical background in Section \ref{sec:mb}.

\medskip

Our results can be summarized as follows.
\def\leftmargini{2.21em}~\\[-1.55em]\begin{itemize}\addtolength{\itemsep}{-7pt}
\item[(1a)]
For a boundary adjacent to a three-dimensional region that is labeled
by a modular tensor category \C, and thus with bulk Wilson lines given by 
\C\ as well, the central information about a topological boundary 
condition $a$ is contained in the process of moving Wilson lines to the 
boundary. It is mathematically described by a central functor
$\Fto a\colon \C\To \calw_a$, with $\calw_a$ the fusion 
category of Wilson lines in the boundary with boundary condition $a$.

\item[(1b)]
A careful distinction between the three-dimensional physics in the 
bulk and the two-di\-men\-sional physics in the boundary allows one
to argue that the functor $\Fto a$ lifts to a (braided) equivalence 
$\tFto a\colon \C\,{\stackrel\simeq\to}\,\Z(\calw_a)$
between the category \C\ of bulk Wilson lines and the
Drinfeld center of the category $\calw_a$ of boundary Wilson lines.

\item[(1c)]
This equivalence implies that a topological boundary condition exists 
for a TFT labeled by the modular tensor category \C\ if and only if the 
class of \C\ in the Witt group of modular tensor categories is trivial.
Put differently, topological boundary conditions exist if and only if the 
modular tensor category \C\ is the Drinfeld center of a fusion category.

\item[(1d)]
For fixed \C, the three-layered structure carried by the boundary 
conditions and their higher-codimension substructures is a bicategory. 
It naturally encodes e.g.\ the fusion of (generalized) Wilson lines.
 \\
This bicategory can be constructed from any single boundary condition 
described by a central functor $\Fto a\colon \C\To \calw_a$ as the 
bicategory of module categories over the fusion category $\calw_a$.
The 1-morphisms of this bicategory -- i.e.\ module functors -- 
describe the possible Wilson lines (one-dimensional defects) on the 
boundary, including their fusion.
The 2-morphisms describe the possible junctions of Wilson lines.

\item[(2)]
A similar analysis can be performed for surface defects separating TFTs
that are labeled by modular tensor categories $\C_1$ and $\C_2$. There 
are now two different processes of moving bulk Wilson lines from either 
$\C_1$ or $\C_2$ into the defect surface with a fusion category
$\calw_d$ of defect Wilson lines. They yield two central functors, 
which can be combined into a braided equivalence 
$\C_1^{} \boti\, \C_2^\rev {\stackrel\simeq\to}\,\Z(\calw_d)$.
Here $\C_2^\rev$ is the modular category with reversed braiding
as compared to $\C_2$, and $\boxtimes$ is the Deligne tensor product. 
Again this equivalence fully captures a surface defect. 
 \\
Thus topological
surface defects exist if and only if $\C_1$ and $\C_2$ are in 
the same Witt class. Again, once one defect is described by an 
equivalence $\C_1^{} \boti \C_2^\rev\,{\stackrel\simeq\to}\,\Z(\calw_d)$
of braided fusion categories, the bicategory of all topological 
surface defects separating $\C_1$ and $\C_2$
is given by the bicategory of $\calw_d$-modules.
\end{itemize}

The description of boundary conditions and surface defects in terms of 
module categories that is achieved in this paper allows for a rigorous 
treatment of related issues. For instance, we can show that all module 
functors appearing in our theory admit ambidextrous adjunctions, which
brings the technology of string diagrams for bicategories to our disposal. 
This way we can e.g.\ provide mathematical foundations for the constructions 
in \cite{kaSau3}; in particular we prove:
\def\leftmargini{2.21em}~\\[-1.55em]\begin{itemize}
\item[(3)]
To every (special) topological surface defect $S$ separating a TFT labeled 
by the modular tensor category \C\ from itself, string diagrams provide,
for any Wilson line separating $S$ and the transparent surface defect,
an explicit construction of a special symmetric Frobenius algebra in \C.
Different Wilson lines give Morita equivalent algebras; we realize the 
Morita context explicitly in terms of string diagrams.
\end{itemize}

\noindent
Before proceeding to the main body of the text, a few further remarks seem
to be in order:
  \def\leftmargini{1.57em}\\[-1.45em]\begin{itemize}\addtolength{\itemsep}{-7pt}
  \item[\nxt]
A TFT of Reshetikhin-Turaev type based on a Drinfeld center of a fusion 
category $\cala$ is, by the results of \cite{balKi,tuVi}, equivalent to 
a TFT of Turaev-Viro type based on $\cala$. Topological boundary conditions
for TFTs of Reshetikhin-Turaev type thus only exist
if the TFT admits a Turaev-Viro type description. 

  \item[\nxt]
Not surprisingly, the description of boundary conditions and defects in 
three-dimensional theories is one step higher in the categorical ladder 
than for \twodim\ theories, e.g.\ two-dimensional CFTs, for which boundary 
conditions and defect lines
form categories of modules and of bimodules, respectively.  
 \\
In fact one expects a relation of boundary conditions for the TFT based 
on the modular tensor category \C\ and \C-module categories. And indeed, 
as we will explain in Sections \ref{sec:bc} and \ref{sec:sd}, respectively, 
the existence of a consistent fusion of bulk and boundary Wilson lines 
requires such a relation. However, \emph{not} every \C-module category 
describes a topological boundary condition. Rather, the structure we present 
involves more stringent requirements that are fulfilled only by a subclass 
of \C-module categories. Analogous comments apply to topological surface defects.

  \item[\nxt]
We describe surface defects and boundary conditions as specific objects of
a bicategory, not just as isomorphism classes thereof. This opens up the 
perspective to obtain a vast extension of the entire Reshetikhin-Turaev 
construction to manifolds with substructures of arbitrary codimension.
Here we will not delve into this issue further, but just mention that
a first inspection indeed 
indicates that one can associate the appropriate vector spaces of 
conformal blocks to cut-and-paste boundaries of such extended manifolds.
Any such construction should respect the known relations between
topological field theories of Turaev-Viro and of Reshetikhin-Turaev
type and therefore be compatible with the kind of construction that is
sketched in \cite{kiKon}.

  \item[\nxt]
We obtain our results separately for boundary conditions and for surface 
defects. A comparison of the results shows
that the two situations are related by a `folding' procedure. We thus 
find a three-dimensional realization of the `folding trick', which 
in \twodim\ conformal field theory is often invoked as a heuristic tool.

\item[\nxt]
We finally comment on surface defects separating \C\ from itself. For
the Deligne product $\C\boti\C^\rev$ of any modular tensor category \C\ 
there exists canonically a braided equivalence to the center of a
fusion category, namely to the center of \C\ itself,
$\C\boti\C^\rev\,{\simeq}\,\Z(\C)$.
Thus there exist topological surface defects separating the
TFT labeled by \C\ from itself. Among them there is in particular the
\emph{transparent}, or invisible, surface defect whose presence is
equivalent to having no interface at all. It corresponds to \C\ seen 
as a module category over itself. The generalized Wilson lines
on the transparent surface defect are just the ordinary Wilson lines.

\end{itemize}

The rest of this paper is organized as follows.
We start by providing some mathematical background information in Section 
\ref{sec:mb}; the reader already familiar with the relevant aspects
of monoidal categories can safely skip this part. Afterwards 
we present details of our proposal for boundary conditions
(Section \ref{sec:bc}) and surface defects (Section \ref{sec:sd}). In 
section \ref{sec:abCS} we then use the relation between module categories 
and Lagrangian algebras to show that, in the specific case of abelian 
Chern-Simons theories, our analysis gives the same results as the 
Lagrangian analysis of \cite{kaSau2}. We conclude in Section 
\ref{sec:Frob} with a model-independent study that extends the results 
of \cite{kaSau3} about the relation between Frobenius algebras in a 
modular tensor category $\C$ and generalized Wilson lines separating 
the transparent surface defect for \C\ from an arbitrary surface defect.


\section{Mathematical preliminaries}\label{sec:mb}

We start by summarizing some pertinent mathematical background. By 
$(\calc,\otimes_\C,\one,a_\C, l_\C,r_\C)$ we denote a monoidal category with 
tensor product $\otimes_\C$, tensor unit $\one$, associativity constraint 
$a_\C$, and left and right unit constraints $l_\C,r_\C$ that obey the 
pentagon and triangle constraints. In our discussion we will, however, 
usually suppress the associativity and unit constraints altogether, as is 
justified by the coherence theorem. We work over a fixed ground
field \ko\ that is algebraically closed and has characteristic zero.
For definiteness we take $\ko$ to be the field $\complex$ of complex numbers,
which is the case relevant for typical applications.
All categories are required to be $\ko$-linear and abelian. 

As we are interested in generalizations of the Reshetikhin-Turaev construction, 
all categories will be finitely semisimple, i.e.\ all objects are
projective, the number of isomorphism classes of simple objects is finite,
and the tensor unit is simple. If such a category is also rigid monoidal
and has finite-dimensional morphism spaces, 
it is called a \emph{fusion category}. With some further structure, such 
categories encode Moore-Seiberg data of chiral conformal field theories. 
(Examples can be constructed from even lattices, see Section \ref{sec:abCS}.) 
We are particularly interested in \emph{braided} categories, i.e.\ 
monoidal categories \C\ endowed with a natural isomorphism from \C\ to
\C\ with the opposite tensor product (i.e.\ with a commutativity constraint)
satisfying the hexagon axioms.

Objects $U,V$ of a braided fusion category are said to \emph{centralize}
each other iff the monodromy $c_{U,V}\cir c_{V,U}$ is the identity morphism. 
For \cald\ a fusion subcategory of a braided fusion category \C, the
\emph{centralizer} $\cald'$ of \cald\ is the full subcategory of objects
of \C\ that centralize every object of \cald.
A braided fusion category is called \emph{non-degenerate} iff
$\C' \,{\simeq}\, \Vectk$ \cite[Def.\,2.1]{dmno}; a 
braided fusion category is called \emph{premodular} iff it 
is equipped with a twist (or, equivalently, with a spherical structure). A 
premodular category is \emph{modular}, i.e.\ its braiding is maximally 
non-symmetric, iff it is non-degenerate \cite[Prop.\,3.7]{dgno2}.


\subsection{Module categories}

Categorification of the standard notion of module over a ring yields the 
notion of a module category over a monoidal category. Similarly, the notion 
of a bimodule category is the categorification of the notion of a bimodule.

\begin{Definition} ~\\[1pt]
(i)\, A (left) \emph{module category} over a monoidal category
$(\cala,\otimes_{\!\cala},\one,a_\cala,l_\cala,r_{\!\cala})$ or, in short, an 
\emph{\cala-module}, is a quadruple $(\calm,\otimes,a,l)$, where \calm\ is a
\ko-linear abelian category and $\otimes\colon \cala\,{\times}\,\calm \To \calm$ 
is an exact bifunctor, while $a \eq (a_{U,V,M})_{U,V\in\cala,M\in\calm}$ 
and $l \eq (l_M)_{M\in\calm}$ are natural families of 
isomorphisms $a_{U,V,M}\colon (U\otia V) \oti M \To U\oti(V\oti M)$
and $l_M\colon \one\oti M \To M$ that satisfy pentagon and triangle 
axioms analogous to those valid for a monoidal category.\,%
 \footnote{~For a complete statement of the axioms see e.g.\
 \cite[Sect.\,2.3]{ostr}.}
\\[2pt]
(ii)\, In the same spirit, 
for $(\cala_1,\otiae,\one_1,a_{\cala_1},l_{\cala_1},r_{\!\cala_1})$
and $(\cala_2,\otiaz,\one_2,a_{\cala_2},l_{\cala_2},r_{\!\cala_2})$ monoidal 
categories, a $\cala_1$-$\cala_2$-\emph{bimodule category}, or
$\cala_1$-$\cala_2$-\emph{bimodule}, is a tuple 
$(\calx,\otimes_1,a_1,l_1,\otimes_2,a_2,r_2,b)$, where \calx\ is a 
\ko-linear abelian category,
  \be
  \otimes_1\colon~~ \cala_1\,{\times}\,\calx \To \calx
  ~\quad\text{ and }\quad~ 
  \otimes_2\colon~~ \calx \,{\times}\,\cala_2 \To \calx
  \ee
are bifunctors, while
$a_1 \eq (a_{1;U,V,X})_{U,V\in\cala_1,X\in\calx}$, 
$l_1 \eq (l_{1;X})_{X\in\calx}$ and
$a_2 \eq (a_{2;X,U,V})_{U,V\in\cala_2,X\in\calx}$,
     \linebreak
$l_2 \eq (l_{2;X})_{X\in\calx}$
as well as $(b_{U;X;V})_{U\in\cala_1,V\in\cala_2X\in\calx}$
are natural families of isomorphisms 
$a_{1;U,V,X}\colon 
    $\linebreak$
(U\otiae\! V) \otie X \To U\otie(V\otie X)$,
$l_{1;X}\colon \one_1 \otie X \To X$,
$a_{2;X,U,V}\colon X \otiz (U\otiaz \!V) \To (X \otiz U) \otiz 
    $\linebreak$
V)$,
$l_{2;X}\colon X \otiz \one_2 \To X$ and
$b_{U;X;V}\colon (U\otie X) \otiz V \To  U\otie (X \otiz V)$
that satisfy pentagon and triangle 
axioms similar to those valid for a monoidal category.\,%
 \footnote{~A complete statement of the axioms can e.g.\ be found in
 \cite[Def.\,2.10\,\&\,Prop.\,2.12]{greeno2}.}
\end{Definition}

\begin{rem} ~\\[1pt]
(i)\, Very much like a ring is a left module over itself, any monoidal 
category \cala\ is naturally a module category over itself;
we denote this `regular' \cala-module by \aa.
Also, via $F\boti A \,{:=}\, F(A)$ every category \cala\ is a module
over the monoidal category $\END(\cala)$ of endofunctors of \cala.
\\[3pt]
(ii)\, Module categories over \cala\ can be described in terms
of algebras in \cala, i.e.\ objects $A$ of \cala\ together
with a multiplication morphism $m\colon A\oti A\To A$ and a unit
morphism $\eta\colon \one\To A$ that obey associativity and unit axioms.
As usual one introduces a category
\Mod-$A$ of right $A$-modules in \cala. One easily verifies
that the functor $(U,M) \,{\mapsto}\, U\oti M$ endows the 
category \Mod-$A$ with the structure of a module category over \cala\
\cite[Sect.\,3.1]{ostr}. Conversely, given a module category, algebras 
can be constructed in terms of internal Homs.
\\[3pt]
(iii)\, Algebras that are not isomorphic can yield equivalent
module categories. In fact, there is a Morita theory generalizing the 
classical Morita theory of algebras over commutative rings.
\\[3pt]
(iv)\, An \cala-module \calm\ is the same as a monoidal functor from 
\cala\ to the monoidal category $\END(\calm)$ of endofunctors of \calm\
\cite[Prop.\,2.2]{ostr}.
\\[3pt]
(v)\, We recall that for our purposes we assume all categories to be abelian 
categories enriched over the category of finite-dimensional complex vector 
spaces and to be finitely semisimple.
\end{rem}

Along with module categories there come 
corresponding notions of functors and natural transformations.

\begin{Definition} ~\\[1pt]
(i)\, A (strong) \emph{module functor} between two \cala-modules $\calm$ and 
$\calm'$ is an additive functor $F\colon \calm \To \calm'$ together with 
a natural family $b \eq (b_{U,M})_{U\in\cala,M\in\calm}$ of isomorphisms 
 \\
$b_{U,M}\colon F(U\oti M) \To U \oti F(M)$ that satisfy pentagon and 
triangle axioms analogous to those valid for a monoidal functor.
\\[2pt]
(ii)\, A \emph{natural transformation} between two module 
functors is a natural transformation of $\ko$-li\-near additive functors
compatible with the module structure.
\\[2pt]
(iii)\, The corresponding notions for bimodule categories are defined
analogously.
\end{Definition}

There is also an obvious operation of \emph{direct sum} of \cala-modules:
$\calm\,{\oplus}\calm'$ is the Cartesian product of 
the categories $\calm$ and $\calm'$ with coordinate-wise 
additive and module structure. An \emph{indecomposable} 
\cala-module is one that is not equivalent (as \cala-modules, 
i.e.\ via a module functor) to a direct sum of two nontrivial \cala-modules. 
Any \cala-module can be written as a direct sum of indecomposable ones, 
uniquely up to equivalence.


\subsection{Bicategories and Deligne products}

Given a monoidal category \cala, the collection of all \cala-modules
has a three-layered structure, consisting of \cala-modules, module functors,
and module natural transformations. This structure cannot be described any
longer in terms of a category; we rather need the notion of a \emph{bicategory},
which is pervasive in this paper. A bicategory has three layers of structure:
objects, 1-morphisms and 2-morphisms. The composition of 1-morphisms is not
necessarily strictly associative, but only up to 2-isomorphisms; if it 
\emph{is} strictly associative, one calls the bicategory \emph{strict} 
(or a 2-category). For 2-morphisms there are two different concatenations, 
referred to as vertical and horizontal compositions.
For details about bicategories see e.g.\ \cite{bena}.

A standard example for a strict bicategory is the one for which objects are
small categories, 1-morphisms are functors and 2-morphisms are natural
transformations. An example for a non-strict bicategory is the one whose
objects are associative algebras, 1-morphisms are bimodules and
2-morphisms are bimodule maps. Here we are interested, for a given
monoidal category \cala, in its bicatgeory \cala-\MOD\ of modules,
having \cala-modules as objects, module functors as 1-morphisms
and natural transformations between module functors as 2-morphisms.
Similarly, for any pair $(\cala_1,\cala_2)$ of monoidal categories there 
is the bicatgeory $\cala_1$-$\cala_2$-\BIMOD.

The universal property of the tensor product of vector spaces 
allows one to describe bilinear maps in terms of linear maps out of the 
tensor product. Similarly, the \emph{Deligne tensor product} 
$\C_1\boxtimes\C_2$ \cite[Sect.\,5]{deli} of abelian categories provides a 
bijection between bifunctors $F\colon \C_1\,{\times}\; \C_2\To \D$ and 
functors $\hat F\colon \C_1\boti \C_2\To \D$. If $\C_1 \eq A_1\mbox-\Mod$ 
is the category of (left, say) modules over a finite-dimensional \ko-algbera 
$A_1$ and $\C_2 \eq A_2\mbox-\Mod$, then $\C_1\,{\boxtimes}\,\C_2$ is 
equivalent to the category of modules over the \ko-algebra 
$A_1\,{\otimes_\ko}\, A_2$ \cite[Prop.\,5.5]{deli}, and if $\C_1$ and $\C_2$ 
are semisimple with simple objects given by $S_i$ and $T_j$, respectively, 
then $\C_1\boti\C_2$ is semisimple as well, with simple objects given 
by $S_i\boti T_j$.

A significant feature of bimodules over a ring is that they admit a 
tensor product. The Deligne product can be used in a similar way. Given, 
say, rings $R_1,R_2$ and $R_3$, the tensor product provides us with functors
  \be
  \otimes_{R_2}:\quad R_1\mbox-R_2\mbox-\Bimod
  \times R_2\mbox-R_3\mbox-\Bimod \to R_1\mbox-R_3\mbox-\Bimod
  \ee
describing `mixed' tensor products. The Deligne tensor product 
categorifies this feature as well and provides bifunctors between 
bimodule categories. For details we refer to \cite[Sect.\,1.46]{egno}.
For a \emph{commutative} ring $R$, the tensor product of two
$R$-modules is again an $R$-module. Braided tensor categories are 
categorifications of commutative rings. Indeed, if \C\ is a \emph{braided} 
abelian monoidal category, then the Deligne tensor product endows the 
bicategory \C-\MOD\ with a monoidal structure. 

Next we notice that for any \ko-algebra $A$ the space $\End_{A}(A_A)$ of 
module endomorphisms of $A$ as a module over itself is isomorphic to 
$\Homk(\ko,A)$ and thus to $A$. This suggests to study the properties of 
the category $\ENDA(\aa)$ of module endofunctors of \cala\ as a module category 
over itself. Since endofunctors can be composed, $\ENDA(\aa)$ is a monoidal 
category. Moreover, we have the following categorified version of the classical 
isomorphism $\End_{A}(A_A) \,{\cong}\, A$ of algebras:

\begin{prop} \label{gibtsdoch}\mbox{} \\
Let \cala\ be a \ko-linear monoidal category. For any object $U\iN\cala$ denote 
by $F_U\colon \aa\To\aa$ the module endofunctor that acts on objects by 
tensoring with $U$ from the left, $F_U(V) \,{:=}\, U\oti V$. Then the functor

  \be
  \begin{array}{rll}
  F_\cala:\quad \cala &\!\longrightarrow\!& \ENDC(\aa)\\[2pt]
  U &\!\longmapsto\!& F_U
  \end{array}
  \label{FC-equiv}\ee
is an equivalence of monoidal categories.
\end{prop}

\begin{proof}
We first show that the functor
  \be \begin{array}{rll}
  G_\cala:\quad \ENDA(\aa) &\!\longrightarrow\!& \cala \\[2pt]
  F &\!\longmapsto\! & F(\one)
  \end{array}
  \ee
is an essential inverse of $F_\cala$. Indeed we have the chain of equalities
$G_\cala \cir \Fcala (U)\eq G_\cala(F_U) \,{=}
    $\linebreak[0]$
 F_U(\one) \eq U\oti\one \eq U$,
so that $G_\cala \cir \Fcala \eq \Id_\cala$. 
Conversely, for any $\varphi \iN \ENDA(\aa)$
the functor $\Fcala \cir G_\cala (\varphi)\iN\ENDA(\aa)$ acts on $U\iN\aa$ as
$(\Fcala \cir G_\cala (\varphi))(U) \eq F_{G_\cala (\varphi)}(U) 
\eq \varphi(\one)\oti U$. The unit constraint of the module functor $\varphi$ 
then provides a natural isomorphism to the identity functor. 
\\[2pt]
It remains to obtain tensoriality constraints for the functor $\Fcala$. The
equalities
  \be
  \bearl
  \Fcala(U{\otimes} V)(W) = (U\oti V)\oti W
  \\[3pt] \hsp{7.2}
  \stackrel{\alpha}\longmapsto\, U\oti(V\oti W) = \Fcala(U)(\Fcala(V)W) =
  (\Fcala(U)\cir \Fcala(V)) (W)
  \eear
  \ee
show that these are afforded by the associativity constraint $a_\cala$ of \cala.
\end{proof}


\subsection{Drinfeld center and enveloping category}
\label{sec:drinfeld}

For algebras over fields, a very useful invariant of the Morita class of an
algebra is the center. In our situation, i.e.\ for algebras in a monoidal 
category \cala, a similar invariant is at hand which
still is an algebra, albeit in a category different from
\cala, namely in the \emph{Drinfeld center} $\Z(\cala)$.
We recall the definition of the Drinfeld center: for \cala\ a monoidal 
category, the objects of the category $\Z(\cala)$ are pairs $(U,e_U)$, 
where $U\iN\C$ and $e_U$ is a `half-braiding', i.e.\ a functorial 
isomorphism $e_U\colon U\oti{-} \,{\stackrel\simeq\longrightarrow}\, {-}\oti U$ 
satisfying appropriate axioms, see e.g.\ \cite[Ch.\,XIII.4]{KAss}. 
$\Z(\cala)$ has a natural structure of
a braided monoidal category. The forgetful functor
  \be
  \begin{array}{rcl}
  \varphi_{\!\cala}:\quad \Z(\cala) &\! \to \!& \cala \\[1pt]
  (U,e_U) &\! \mapsto \!& U
  \eear
  \label{Rvarphi}\ee
is a tensor functor.

The \emph{reverse} category of a braided monoidal category $\C$, denoted by 
$\C^\rev$, is the same category with opposite braiding; if \C\ is even a 
ribbon category, as in all our applications, we also endow it with the opposite 
twist. The Deligne product 
  \be
  \Ce := \C \boti \C^\rev
  \ee 
is a categorified version of the enveloping algebra 
$A^\text e \eq A \otik A^\text{op}$ of an associative
algebra. Accordingly we call \Ce\ the \emph{enveloping category} of \C.
And in the same way as the category of $A$-bimodules can be described, as an
abelian category, in terms of $A^\text e$-modules, the bicategory
$\C_1$-$\C_2$-\BIMOD\ is equivalent to the bicategory
$(\C_1^{}{\boxtimes}\C_2^\rev)$-\MOD.

Suppose now that the monoidal category \C\ is already braided itself, with 
braiding $c$. Then the braiding provides a functor, actually a braided tensor 
functor, from \C\ into its center $\Z(\C)$ by $U \,{\mapsto}\, (U,c_{U,-})$.
We also have a braided tensor  functor $\C^\rev\To \Z(\C)$, which is obtained 
by the opposite braiding: $U\,{\mapsto}\, (U,c_{-,U}^{-1})$. 
Using the universal property of the Deligne tensor product we combine 
the two functors into a tensor functor
  \be 
  \begin{array}{rll}
  G_\C:\qquad \Ce &\!\longrightarrow\!& \Z(\C)\\[2pt]
  U \boxtimes V &\!\longmapsto\!& (U\oti V,e_{U\otimes V}) 
  \end{array}
  \labl{GC}
where 
  \be
 \xymatrix{
  e_{U\otimes V}(W)^{}:\quad
  U\oti V\oti W \ar^{\,~~~~~~~~\idsm_U\otimes c_{W,V}^{-1}}[rr]
  && ~U\oti W\oti V \ar^{\,c_{U,W} \otimes \idsm_V}[rr]
  && W\oti U\oti V .
  } \ee
The functor $G_\C$ has a natural structure of a braided tensor functor. A 
braided monoidal category is called \emph{factorizable} iff $G_\C$ is an 
equivalence of braided monoidal categories. Representation categories of 
finite-dimensional factorizable Hopf algebras in the sense of 
\cite{drin10} are factorizable. 

It is natural to ask under what condition the functor $G_\C$ is a braided
equivalence. This is answered by the

\begin{lem}\label{defi:modular}{\rm\cite{muge9,etno}}\\
For \C\ a semisimple ribbon category, the functor $G_\C$ \eqref{GC} is 
an equivalence between the center $\Z(\C)$ and the enveloping category \Ce\
if and only if \,\C\ is a modular tensor category.
\end{lem} 

Thus in particular in the context of the Reshetikhin-Turaev construction,
which takes as an input a \emph{modular} tensor category,
the center and enveloping category of \C\ are equivalent as
braided categories, including their spherical structure.

For a braided category \C\ the obvious functor $\Ce\To\C$ factors through 
the center of \C: composing the functor $G_\C$ \eqref{GC} with the 
forgetful functor, we obtain
  \be
  \Ce \to \Z(\C) \to \C \,.
  \ee
Hereby \C\ becomes a \Ce-module, and any \C-module is turned into a 
\C-bimodule.

The following assertion shows again that it is appropriate to regard
the Drinfeld center as a categorification of the center of an algebra:

\begin{prop}{\rm \cite[Thm.\,3.1]{etno3}, \cite[Rem.\,3.18]{muge8}}
\label{2Morita}\\
Let \cala\ and \calb\ be fusion categories. Their centers $\Z(\cala)$ and 
$\Z(\calb)$ are braided equivalent iff their bicategories \cala-\MOD\ and 
\calb-\MOD\ of module categories are equivalent.
\end{prop}

There is a close relation between module categories and the Drinfeld center 
\cite[Sect.\,5.1]{enoM}. For any indecomposable \cala-module 
\calm\ over a fusion category \cala, the category $\ENDA(\calm)$ of
\cala-module endofunctors of \calm\ is a fusion category, and \calm\ can be 
regarded as a right $\ENDA(\calm)$-module, and thus as an
$\cala{\boxtimes}\ENDA(\calm)^\rev$-module. The 
$\cala{\boxtimes}\ENDA(\calm)^\rev$-module endofunctors of this module
category can be identified \cite[Sect.\,2.6]{dmno} with the functors of 
tensoring with an object of the Drinfeld center $\Z(\cala)$ from the left, 
or, alternatively, with the functors of tensoring with an object of 
$\Z(\ENDA(\calm))$ from the right. 
Comparing the two descriptions of these functors gives the following result:

\begin{prop}\label{prop:schau13}{\rm\cite{schau13}}\\
For any module \calm\ over a fusion category \cala\
there is a canonical equivalence
  \be\label{braideq}
  \Z(\cala) \stackrel\simeq\longrightarrow \Z(\ENDA(\calm))
  \ee
of braided categories.
\end{prop}


\subsection{Central functors}

In this brief subsection, we recall a notion that will enter
crucially into our analysis of boundary conditions and defect
surfaces.

\begin{Definition}\label{def:centfun}\cite[Sect.\,2.1]{bezr} \\[1pt]
A structure of a \emph{central functor} on a monoidal functor 
$F\colon \C\To\cala$
from a braided monoidal category \C\ to a monoidal category \cala\ is
a natural family of isomorphisms
  \be
  \sigma_{U,V}^{}:\quad F(U)\oti V \,{\stackrel\cong\longrightarrow}\,V\oti F(U)
  \ee 
for $U$ in \C\ and
$V$ in \cala, satisfying the following compatibility conditions:
 \\[2pt]
(i)\, For $X,X'\iN \C$ the isomorphism $\sigma_{X,F(X')}$ coincides with
the composition
  \be
  F(X)\otimes F(X') \cong F(X\oti X') \cong F(X'\oti X)
  \cong F(X')\otimes F(X) \,,
  \ee
where the first and the third isomorphisms are the tensoriality constraints 
of $F$, while the middle isomorphism comes from the braiding on \C.
 \\[2pt]
(ii)\, For $Y_1, Y_2\iN \cala$ and $X\iN \C$ the composition
  \be
  \xymatrix{
  F(X)\otimes Y_1\otimes Y_2 \ar^{\sigma_{X,Y_1}^{}\otimes Y_2~}[rr]
  && Y_1 \otimes F(X) \otimes Y_2 \ar^{Y_1\otimes \sigma_{X,Y_2}^{}~}[rr]
  && Y_1\otimes Y_2\otimes F(X)
  } \ee
coincides with the isomorphism $\sigma_{X, Y_1\otimes Y_2}$.
 \\[2pt]
(iii)\, For $Y\in \cala$ and $X_1,X_2\in \C$ the composition
  \be
  \begin{array}{r}
  \xymatrix{
  F(X_1\oti X_2)\otimes Y \,\cong\, F(X_1) \otimes F(X_2) \otimes Y
  \ar^{\,~~~~~~~~~~~~~F(X_1)\otimes\sigma_{X_2,Y}^{}}[rr]
  && ~F(X_1) \otimes Y \otimes F(X_2)
  }
  \qquad\qquad \\~\\[-.8em] \xymatrix{
  {~} \ar^{\sigma_{X_1,Y}^{}\otimes F(X_2)~~~~~~~~~~~~~~~~~~~~~~~~~~~}[rr]
  && ~Y\otimes F(X_1)\otimes F(X_2) \,\cong\, Y \otimes F(X_1\oti X_2)
  } \end{array}
  \ee
coincides with $\sigma_{X_1\otimes X_2, Y}$.
\end{Definition}

The following result relates central functors
into \cala\ to the Drinfeld center $\Z(\cala)$:

\begin{lemma}\label{lem:central}{\rm\cite[Def.\,2.4]{dmno}} \\[1pt]
A structure of central functor on $F\colon \C\To\cala$ is equivalent to
a lift of $F$ to a braided tensor functor 
$\widetilde F\colon \C \To \Z(\cala)$, i.e.\ the composition 
$\varphi_{\!\cala} \cir \widetilde F$ with the forgetful functor
\eqref{Rvarphi} equals $F$,
  \be
  \xymatrix{
  & \Z(\cala)  \ar^{\varphi_{\!\cala}^{}}[d]\\
  \C\ar_F[r]\ar@{-->}^{\widetilde F}[ur]& \cala
  }\ee
\end{lemma}


\subsection{Lagrangian algebras}

In general, for an algebra $A$ in a fusion category \cala\ there is
no notion of a center, at least not as an object of \cala. This is simply
because \cala\ is not required to be braided, so that there is no natural 
concept of commuting factors in a tensor product.
As it turns out, the Drinfeld center $\Z(\cala)$, which \emph{is} 
braided, is the right recipient for a notion of a 
center. Keeping in mind that, in classical algebra, Morita equivalent
algebras have isomorphic centers, a center should better be associated 
to a module category over \cala\ rather than to an algebra in \cala.

\begin{Definition} \cite[Defs.\,3.1\,\&\,4.6]{dmno}~\\[1pt]
(i)\, An algebra $A$ in a monoidal category is called \emph{separable} iff
the multiplication morphism splits as a morphism of $A$-bimodules.
\\[3pt]
(ii)\, An algebra in a monoidal category that is also a coalgebra is called 
\emph{special} iff it is separable, with the right-inverse of the product 
given by a multiple of the coproduct, and the composition $\eps\cir\eta$ 
of the counit and unit is non-zero.
\\[3pt]
(iii)\, An \emph{\'etale} algebra in a braided \ko-linear monoidal category 
\C\ is a separable commutative algebra in \C.
\\[3pt]
(iv)\, An \'etale algebra $A\iN\C$ is said to be \emph{connected} 
(or \emph{haploid}) iff $\dimk \Hom(\one,A) \eq 1$.
\\[3pt]
(v)\, A \emph{Lagrangian algebra} in a non-degenerate braided fusion 
category \C\ is a connected \'etale algebra $L$ in \C\ for which the 
category $\C_L^{0}$ of local $L$-modules in \C\ is equivalent to $\Vectk$
as an abelian category.
\end{Definition}

\begin{rem}
(i)\, A local (or dyslectic) module $(M,\rho)$ over a commutative algebra $A$
is an $A$-module for which the representation morphism $\rho$ satisfies
$\rho \cir c_{A,M}^{}\cir c_{M,A}^{} \eq \rho$ \cite{pare23,kios,ffrs}. 
The full subcategory of dyslectic modules is a braided monoidal category.
\\[2pt]
(ii)\, The defining property 
$\C_L^{0} \,{\simeq}\, \Vectk$ of a Lagrangian algebra is equivalent
to the equality $(\fpdim(L))^2 \eq \fpdim(\C)$ of
Perron-Frobenius dimensions \cite[Cor.\,3.32]{dmno}.
\end{rem}

\begin{prop}\label{prop:triple}{\rm\cite[Cor.\,3.8]{danO}}\\
For \C\ a non-degenerate braided fusion category, equivalence
classes of indecomposable \C-modules are in bijection with isomorphism classes 
of triples $(A_1, A_2, \Psi)$ with $A_1$ and $A_2$ connected \'etale algebras
in \C\ and $\Psi\colon \C_{A_1}^0 {\stackrel\simeq\longrightarrow}\, 
\C_{A_2}^{0\;\text{rev}}$ a braided equivalence between the category of 
local $A_1$-mo\-dules and the reverse of the category of local $A_2$-mo\-dules
\end{prop}

\begin{rem}
\'Etale algebras can be obtained from central functors and,
conversely, central functors from induction along \'etale functors:
\\[2pt]
(i)\, Given a central functor $F\colon \C\To\cala$ from a braided fusion 
category \C\ to a fusion category \cala, denote by $\mathrm R_F$ its right 
adjoint functor. The object $\mathrm R_F(\one_\cala)$ then has a canonical 
structure of connected \'etale algebra in \C\ \cite[Lemma\,3.5]{dmno}.
\\[2pt]
(ii)\, For \C\ a braided fusion category and $A$ a connected \'etale algebra 
in \C, the induction functor $\Ind_A\colon \C \To \C_A$ that acts as
$U \,{\mapsto}\, U\oti A$ admits a natural structure of a central functor
\cite[Sect.\,3.4]{dmno}.
\\[2pt]
(iii)\, If in addition \C\ is non-degenerate and $A$ is Lagrangian, then 
the lift $\widetilde{\Ind_A}\colon \C \To \Z(\C_A)$ of the induction functor 
is a braided tensor equivalence \cite[Cor.\,4.1(i)]{dmno}.
\end{rem}

We are now in a position to relate indecomposable module categories over a 
fusion category \cala\ and Lagrangian algebras in its center $\Z(\cala)$. 
Denote by 
  \be
  F:\quad \Z(\cala)\stackrel\simeq\longrightarrow
  \Z(\ENDA(\calm)) \stackrel{\varphi}\longrightarrow \ENDA(\calm)
  \ee
the composition of the equivalence \eqref{braideq} with the forgetful functor.
This is, trivially, a central functor, and the image $A_\calm$ of the tensor 
unit of the monoidal category $\ENDA(\calm)$ 
under the functor $\mathrm R_F$ right adjoint to $F$ is an \'etale
algebra and, as it turns out, even a Lagrangian algebra.

The following proposition shows that these
Lagrangian algebras can be seen as invariants of indecomposable
tensor categories.

\begin{prop} \label{correspondence} {\rm \cite[Prop.\,4.8]{dmno}}\\
For any fusion category $\mathcal{A}$ there is a bijection between the sets 
of isomorphism classes of Lagrangian algebras in $\Z(\cala)$ and 
equivalence classes of indecomposable $\mathcal{A}$-modules.
\end{prop}

\noindent
The proof of this statement is based on Proposition \ref{2Morita}.


\subsection{The Witt group}

One step in the long-standing problem of classifying rational conformal 
field theories is the classification of modular tensor categories.
Recently, the following algebraic structure was established in
the wider context of non-degenerate braided fusion categories
(i.e.\ without assuming a spherical structure):
The quotient of the monoid (with respect to the Deligne product)
of non-degenerate braided fusion categories by its submonoid of Drinfeld 
centers forms a group that
contains as a subgroup the group $\mathfrak W_{\rm pt}$ of the classes of 
non-degenerate pointed braided fusion categories \cite[Sect.\,5.3]{dmno}.
The latter coincides with the classical Witt group \cite{witT1} of metric 
groups, i.e.\ of finite abelian groups equipped with a non-degenerate 
quadratic form. This motivates the

\begin{defi} \cite[Defs.\,5.1\,\&\,5.5]{dmno}\\[1pt]
(i)\, Two non-degenerate braided fusion categories $\C_1$ and $\C_2$ are
called \emph{Witt equivalent} iff there exists a braided equivalence 
$\C_1 \boti \Z(\cala_1) \,{\simeq}\,\, \C_2 \boti \Z(\cala_2)$ with
suitable fusion categories $\cala_1$ and $\cala_2$.
\\[2pt]
(ii)\, The \emph{Witt group} $\mathfrak W$ is the group of Witt equivalence 
classes of non-degenerate braided fusion categories.
\end{defi}

It is not hard to see \cite{dmno} that Witt equivalence is indeed an 
equivalence relation, and that $\mathfrak W$ is indeed an abelian group,
with multiplication induced by the Deligne product. The neutral element 
of $\mathfrak W$ is the class of all Drinfeld centers, and the inverse 
of the class of \C\ is the class of its reverse category $\C^\rev$.

As we will see below, in our considerations the Witt group $\mathfrak W$ 
will play an important role. But we will also be interested in the categories
themselves rather than in their classes in $\mathfrak W$. Moreover, in
our context, the categories whose Witt classes are relevant are even modular.
Accordingly we set:

\begin{defi}\label{def:witt-trv}~\\[1pt]
(i)\, A modular tensor category \C\ is called \emph{Witt-trivial} iff its 
class in the Witt group $\mathfrak W$ is the neutral element of $\mathfrak W$.
\\[2pt]
(ii)\, A \emph{Witt-trivialization} of a modular tensor
category \C\ consists of a fusion category \cala\ and an equivalence
  \be
  \alpha:\quad \C \to \Z(\cala)
  \ee
as ribbon categories.
\end{defi}


\section{Bicategories for boundary conditions}\label{sec:bc}

We are now ready to formulate our proposal for topological boundary
conditions for Reshetikhin-Turaev type topological field theories.
Since a topological field theory of Turaev-Viro type based on
a fusion category \cala\ has a natural description as a
TFT of Reshetikhin-Turaev type based on the Drinfeld center
$\Z(\cala)$, our results cover TFTs of Turaev-Viro type as well.

Recall from the Introduction that the boundary conditions we are going 
to discuss refer to boundaries at which the 
three-dimensional world ends, rather than cut-and-paste boundaries.
As we are working in the Reshetikhin-Turaev framework, in which the
categories labeling three-dimensional regions are modular
categories and thus in particular finitely semisimple, we only allow
for boundary conditions that correspond to finitely semisimple categories 
as well (though not modular and not even braided,
in general, as in two dimensions there is no room for a braiding).

We seize from \cite[Sect.\,5.2]{kaSau2} the idea to analyze what happens
when Wilson lines in the bulk approach the boundary. 
We assume that for a given TFT in the bulk there exists a topological
boundary condition $a$ at the end of the three-dimensional world. 
The two-dimensional boundary can contain Wilson lines. These
Wilson lines can carry insertions, and for this reason they are
labeled by the objects of a category \Wilsa. Boundary Wilson lines can 
be fused, and accordingly \Wilsa\ has the structure of a monoidal category, 
and moreover, owing to the fact that the Wilson lines are topological, 
this comes with dualities.
On the other hand, this category is not braided, in general, since there 
does not exist a natural way to `switch' two boundary Wilson lines without
leaving the boundary which is two-dimensional.

However, there are Wilson lines in the nearby bulk as well; they are labeled 
by some modular tensor category \C\ (the same that labels the bulk region
adjacent to the boundary). The category of bulk Wilson lines is 
in particular braided, since Wilson lines can be switched in the 
three-dimensional region. Now part of what is to be meant by a
boundary condition is to be able to tell what happens when the boundary is
approached from the bulk. Thus we postulate that for a consistent boundary 
condition there should exist a process of adiabatically moving Wilson lines 
in the bulk to the boundary, whereby they turn into boundary Wilson lines. 
Put differently, we postulate that there is a functor
  \be
  \Fto a:\quad \C \to \Wilsa \,.
  \labl{Fbulktoa}
Furthermore, the following two processes should yield equivalent results:
On the one hand, first fusing two bulk Wilson lines in the bulk and then
bringing the so obtained single bulk Wilson line to the boundary; and on
the other hand, first moving the two bulk Wilson lines separately to the
boundary and then fusing them as boundary Wilson lines inside the boundary.
Schematically, showing a two-dimensional section perpendicular to the boundary,
the situation looks as follows:
  \be
  \parbox{100pt}{ \xy
  (20,20)*{\includegraphics[scale=0.4]{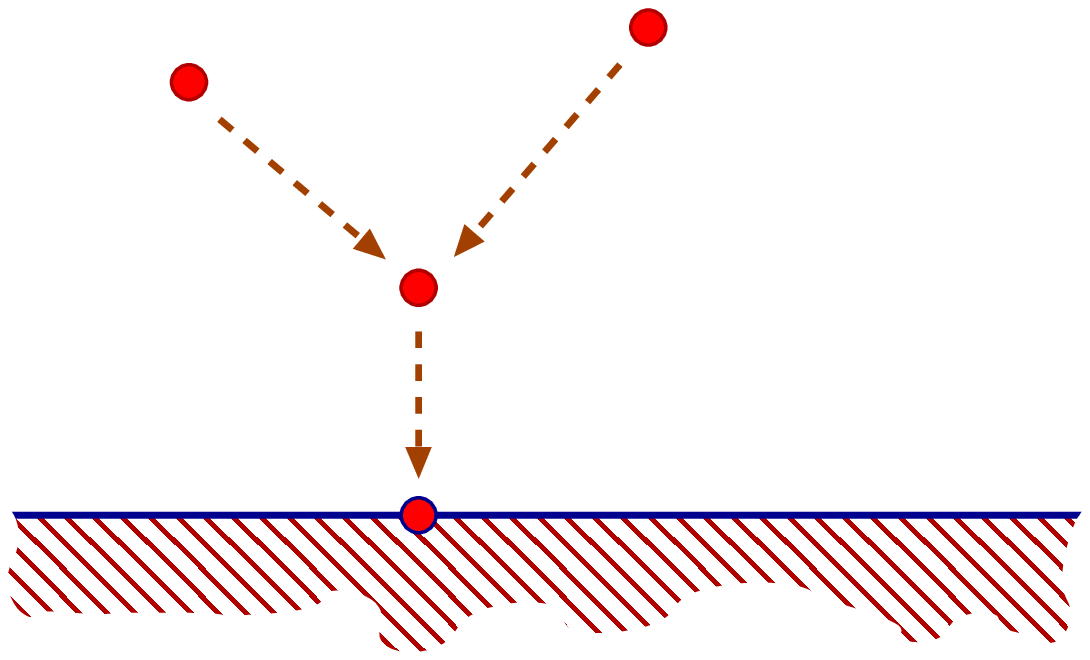}};
  (3.0,32.2)*{\scriptstyle U};
  (19.6,23.2)*{\scriptstyle U{\otimes}V};
  (22.9,17.4)*{\begin{turn}{15}$\scriptstyle \Fto a(U{\otimes}V)$\end{turn}};
  (26.1,34.2)*{\scriptstyle V};
  (90,20)*{\includegraphics[scale=0.4]{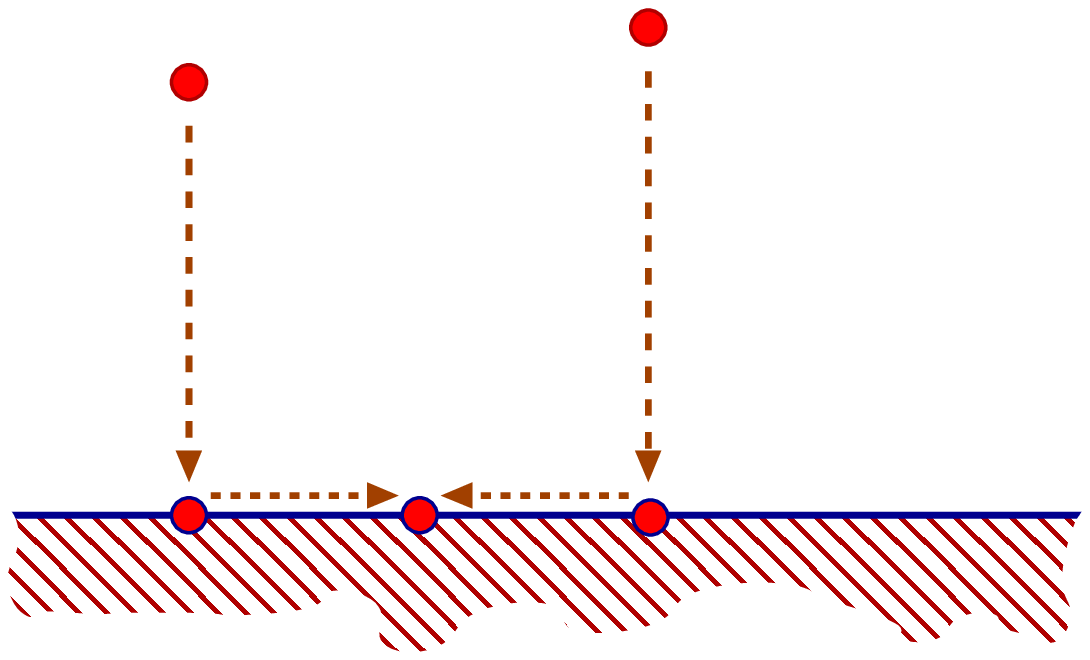}};
  (69.0,16.5)*{\scriptstyle \Fto a(U)};
  (73.0,32.2)*{\scriptstyle U};
  (86.6,27.1)*{\begin{turn}{80}$\scriptstyle \Fto a(U)\oti \Fto a(V)$\end{turn}};
  (96.3,34.4)*{\scriptstyle V};
  (100.3,16.5)*{\scriptstyle \Fto a(V)};
  \endxy }
  \ee
Put differently, the functor \eqref{Fbulktoa} obeys
  \be
  \Fto a(U \oti V) \cong \Fto a(U) \oti \Fto a(V) \,,
  \ee
with coherent isomorphisms, i.e.\ the functor $\Fto a$
has the structure of a tensor functor. From multiple fusion, one
concludes the existence of associativity constraints.
Moreover, we should get the same result when homotopies are applied to 
Wilson lines in the boundary as when they are applied in the bulk. 
Put differently, the functor $\Fto a$ should respect dualities.

The next consideration shows that $F_{\text{bulk}\to a}$ has even more
structure. Consider again the situation that a bulk Wilson line $U \iN \C$ 
is moved to the boundary, resulting in a boundary Wilson line 
$\Fto a(U) \iN \Wilsa$.
Assume in addition that nearby on the boundary there is already another
parallel boundary Wilson line $M \iN \Wilsa$. Since the process of moving $U$ to
the boundary is supposed to be adiabatic, we should get isomorphic results when
we either move $U$ to the \emph{left} of $M$ and then fuse $\Fto a(U)$ with $M$,
or else move $U$ to the \emph{right} of $M$ and then fuse $\Fto a(U)$ with $M$,
as indicated in the following picture.
  \be
  \parbox{100pt}{ \xy
  (20,20)*{\includegraphics[scale=0.4]{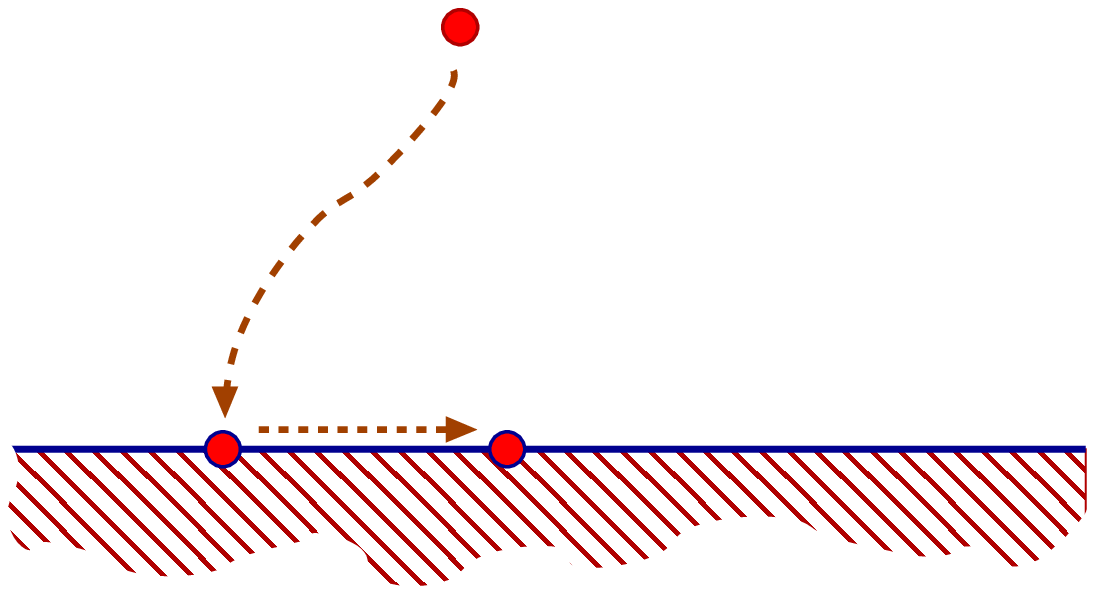}};
  (0.9,18.7)*{\begin{turn}{-20}$\scriptstyle \Fto a(U)$\end{turn}};
  (19.0,32.8)*{\scriptstyle U};
  (20.1,17.9)*{\scriptstyle M};
  (90,20)*{\includegraphics[scale=0.4]{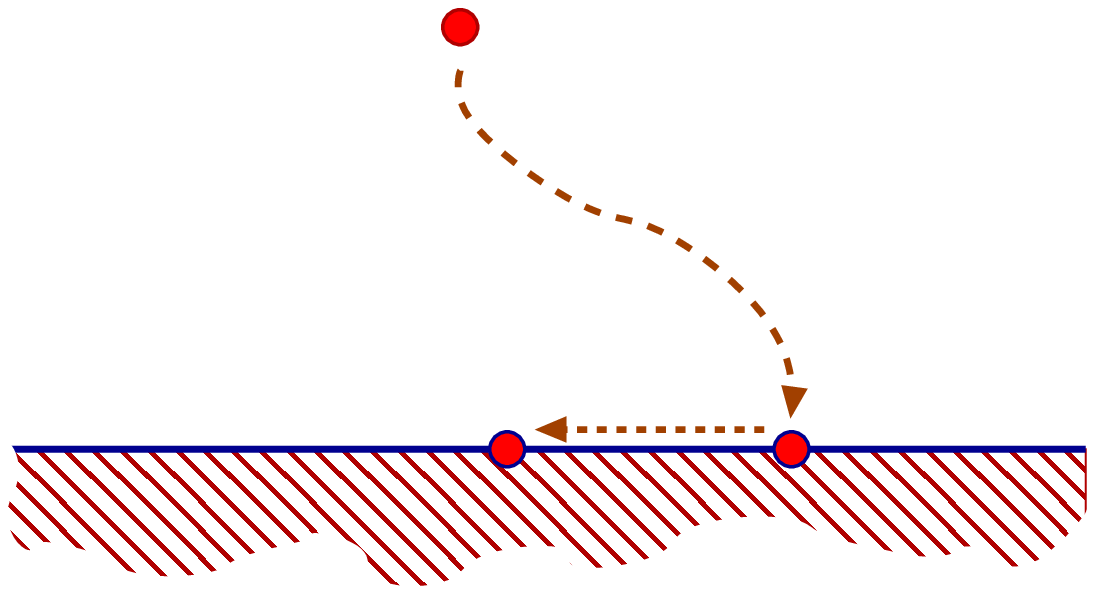}};
  (89.0,32.8)*{\scriptstyle U};
  (88.5,18.0)*{\scriptstyle M};
  (105.9,18.7)*{\begin{turn}{20}$\scriptstyle \Fto a(U)$\end{turn}};
  \endxy }
  \ee
Put differently, we expect a natural isomorphism
  \be
  F_{\to a}(U) \oti M \stackrel\cong\to M \oti F_{\to a}(U) \,.
  \ee
The following argument shows that these isomorphisms endow the functor
$\Fto a$ with the structure of a central functor in the sense of Definition 
\ref{def:centfun}. Property (i) of a central functor is the statement that
for a boundary Wilson line that has been obtained by the adiabatic process, 
the interchange with another such Wilson line comes from the braiding of 
bulk Wilson lines. Property (ii) of a central functor, which may be
called a boundary Yang-Baxter property, is a consequence of the homotopy 
equivalence of two different processes in the bulk: either moving the 
Wilson line in a single step past two boundary Wilson lines, or else
doing it in two separate steps. Property (iii) is seen similarly,
this time with two bulk Wilson lines involved.

According to lemma \ref{lem:central} such a structure is, in turn,
equivalent to a lift of $F_{\text{bulk}\to a}$ to a braided functor
  \be
  \tFto a:\quad \C \to \Z(\Wilsa)
  \labl{tFbulktoa}
from \C\ to the Drinfeld center of the fusion category \Wilsa.

The two-dimensional physics of the boundary surface does not provide any
natural reason for such a half-braiding rule to exist. The only possible 
natural origin of such a rule is thus that it is related to the half-braiding 
in the three-dimensional bulk, via
the processes encoded in the functor $\Fto a$. Accordingly there
should not exist any systematic rule of moving a boundary Wilson line $M$ 
to the other side of a neighbouring boundary Wilson line, except through
the fact that $M$ secretly is a bulk Wilson line that has been brought to the
boundary (so that the rule comes from the process of first bringing it again
into the bulk, moving it around there, and then moving it back to the boundary).

A boundary Wilson line is labeled by an object of $\Wilsa$; a systematic
rule of moving a boundary Wilson line $M$ to the other side of a neighbouring 
boundary Wilson line constitutes a half-braiding $c_{M,-}$ on $\Wilsa$ for the 
object $M$. The pair $(M,c_{M,-})$ is thus just an object in the Drinfeld
center $\Z(\Wilsa)$. Put differently, 
the functor \eqref{tFbulktoa} is essentially surjective.

Similarly, no information about the bulk should be lost when a bulk Wilson 
line is brought to the boundary, provided one remembers the way the Wilson 
line can wander within the bulk to the other side of any other boundary 
Wilson line. This principle applies likewise to insertions on the Wilson 
lines. In the bulk, such insertions are morphisms in $\calc$; for boundary 
Wilson lines, we can only allow morphisms that are compatible with the 
rule to switch the boundary Wilson line with any other boundary Wilson 
line. In other words, we only allow those morphisms of $\Wilsa$ that are 
compatible with the half-braiding, i.e.\ we only consider morphisms in 
$\Z(\Wilsa)$. Put differently, the functor $\tFto a$ \eqref{tFbulktoa} is 
fully faithful, and thus, being also essentially surjective, it is a 
braided \emph{equivalence}:
  \be
  \tFto a:\quad \C \stackrel\simeq\longrightarrow \Z(\Wilsa) \,,
  \label{tFto}
  \ee
{}From this equivalence, we conclude that boundary conditions of the type
we consider can only exist for a bulk theory that is Witt-trivial in the
sense of definition \ref{def:witt-trv}. Even more,
the boundary data are given by a Witt-trivialization.

Once one has understood one boundary condition in a physical theory,
frequently the way is open to understand other boundary conditions as well.
Thus let us assume that there exists another
boundary condition besides $a$. At this point we do \emph{not}, however,
assume that this boundary condition $b$ comes with a central functor
$F_{\to b}$ as well, but rather perform an analysis purely within the boundary.
Consider a generalized Wilson line that separates the boundary condition
$a$ on the left from $b$ on the right. Such Wilson lines can carry local 
field insertions as well; hence we describe them in terms of a category 
$\calw_{a,b}$. We can fuse such a Wilson line with a Wilson line from 
$\Wilsa$ to the left of it. This gives again a Wilson line separating 
the boundary condition $a$ from $b$ and thus an object in $\calw_{a,b}$.
We thus get on the category $\calw_{a,b}$ the structure of a module category 
over $\calw_a$. 

By a similar argument, the category $\calw_b$ of boundary Wilson lines 
separating the boundary condition $b$ from itself has to act on the 
Wilson lines in $\calw_{a,b}$ from the right. Put differently, $\calw_{a,b}$ 
is a right $\calw_b$-module. On the other hand, $\calw_{a,b}$ is already 
naturally a right module category over the category 
$\calw_{a,b}^* \eq \END_{\calw_a}(\calw_{a,b})$
of module endofunctors. We now invoke a principle of naturality
and require that this category describes the tensor category
of generalized boundary Wilson lines for the boundary condition $b$,
i.e.\ that $\calw_b \,{\simeq}\, \calw_{a,b}^*$.

The latter postulate can only make sense if the fusion category 
$\calw_{a,b}^*$ comes with a Witt-trivialization of the bulk category 
$\calc$ as well, i.e.\ if we have a canonical equivalence 
  \be
  \calc \simeq \Z(\calw_{a,b}^*) 
  \label{3.9}  \ee
of braided categories. According to Proposition \ref{prop:schau13}
this is indeed the case. This can be seen as a justification of our 
naturality principle by which we identified Wilson lines with module functors.

To obtain another check of our proposal,
we next consider a trivalent vertex in a boundary, with one incoming
bulk Wilson line labeled by $U\iN\C$, one incoming boundary Wilson 
line labeled by $W_1\iN \calw_a$ and one outgoing boundary Wilson line 
$W_2\iN\calw_a$.  According to our general picture the three-valent 
vertex should be labeled by an element of a vector space obtainable as a 
morphism space. We can realize this vector space in terms of morphisms 
in the category $\calw_a$, provided that there is a mixed tensor product
  \be
  \C \times \calw_a \longrightarrow \calw_a
  \ee
and take the trivalent vertex to be labeled by an element of
$\Hom_{\calw_a}(U\oti W_1,W_2)$. Put differently,
we need the structure of a $\C$-module on $\calw_a$.
To determine what module category is relevant, we invoke topological 
invariance of the bulk Wilson line so as to have it running parallel 
with the boundary before it enters the vertex. We then apply the
adiabatic process described by the functor $\Fto a$ to the piece 
parallel to the surface, thereby turning the bulk Wilson line with 
label $U\iN\C$ into a boundary Wilson line with label $\Fto a(U)$. 
This way we reduce the problem of a trivalent vertex involving a 
bulk Wilson line to the one of a trivalent vertex involving only boundary 
Wilson lines. The relevant vector space is thus $\Hom_{\calw_a}
(\Fto a(U)\otimes W_1,W_2)$. Put differently, we use the $\C$-module 
structure on $\calw_a$ that is induced by pullback of the regular 
module category along the monoidal functor $\Fto a:\C\To\calw_a$
or, what is the same, along the monoidal functor
  \be
 \xymatrix{
  \C \ar^{\tFto a~~~~}[r] 
  & \Z(\calw_a) \ar^{~~\varphi_{\!\calw_a}^{}}[r]
  & \calw_a
  } \ee
from \C\ to $\calw_a$. 

It is important to note that this way one does \emph{not} obtain \emph{all} 
\C-modules; thus our results lead to a selection principle that singles out 
an interesting subclass of \C-modules.
This can be seen already in simple examples, e.g.\ when $\calw_a$ is the
category of finite-dimensional representations of a finite group $G$,
so that $\Z(\cala)$ is the category of finite-dimensional representations of 
the double $D(G)$ \cite[Thm.\,3.1]{ostr5}. For instance, for $G\eq \zet_2$,
there are two indecomposable \cala-modules  (called `rough' and `smooth' in 
\cite{kiKon}), but six indecomposable $\Z(\cala)$-modules.

What we have managed so far is to use one given topological boundary condition
to obtain also other topological boundary conditions. This raises
the question of whether we can obtain \emph{all} topological boundary conditions 
this way. Suppose we are given two different boundary conditions and thus two 
braided equivalences
  \be
  \C\stackrel\simeq\longrightarrow \Z(\cala_1) \quad\text{ and }\quad
  \C\stackrel\simeq\longrightarrow \Z(\cala_2) \,. 
  \ee
Then we have a braided equivalence $\Z(\cala_1)\simeq\Z(\cala_2)$. By the 
result of \cite{etno3} on 2-Morita theory that we recalled in Proposition 
\ref{2Morita}, this implies that the bicategories of $\cala_1$-modules and of 
$\cala_2$-modules are equivalent bicategories. We thus conclude that we can
indeed access every boundary condition from any other boundary condition.

\medskip

We summarize our proposal: Topological boundary conditions for a
topological field theory of Reshetikhin-Turaev type, based on a modular 
tensor category \C, are described by Witt-tri\-vializations of \C,
i.e.\ by braided equivalences $\C\,{\stackrel\simeq\to}\,\Z(\cala)$. Given 
any such trivialization, the bicategory of topological boundary conditions 
can be identified with the bicategory of $\cala$-modules.

One should also appreciate that if a TFT of Turaev-Viro type based on the
fusion category \cala\ is described as a Reshetikhin-Turaev theory based 
on the modular tensor category $\Z(\cala)$, then it comes with a 
trivialization and the category of topological boundary conditions is 
naturally identified with the bicategory of \cala-modules.
In the special case of TFTs of Turaev-Viro type, our results thus reproduce
results of \cite{kiKon} about boundary conditions in such TFTs.


\section{Bicategories for surface defects}\label{sec:sd}

Next we study what kind of mathematical objects describe topological surface 
defects, i.e.\ the topological surface operators considered for abelian 
Chern-Simons theories in \cite{kaSau2} or the domain walls in 
\cite{besW,kiKon}. We consider a surface defect $d$ 
separating two modular tensor categories $\C_1$ and $\C_2$ and follow the 
same line of arguments as for boundary conditions in Section \ref{sec:bc}.
The situation to be studied is displayed schematically in the following picture,
which shows a two-dimensional section perpendicular to the defect surface:
  \be
  \parbox{100pt}{ \xy
  (20,20)*{\includegraphics[scale=0.4]{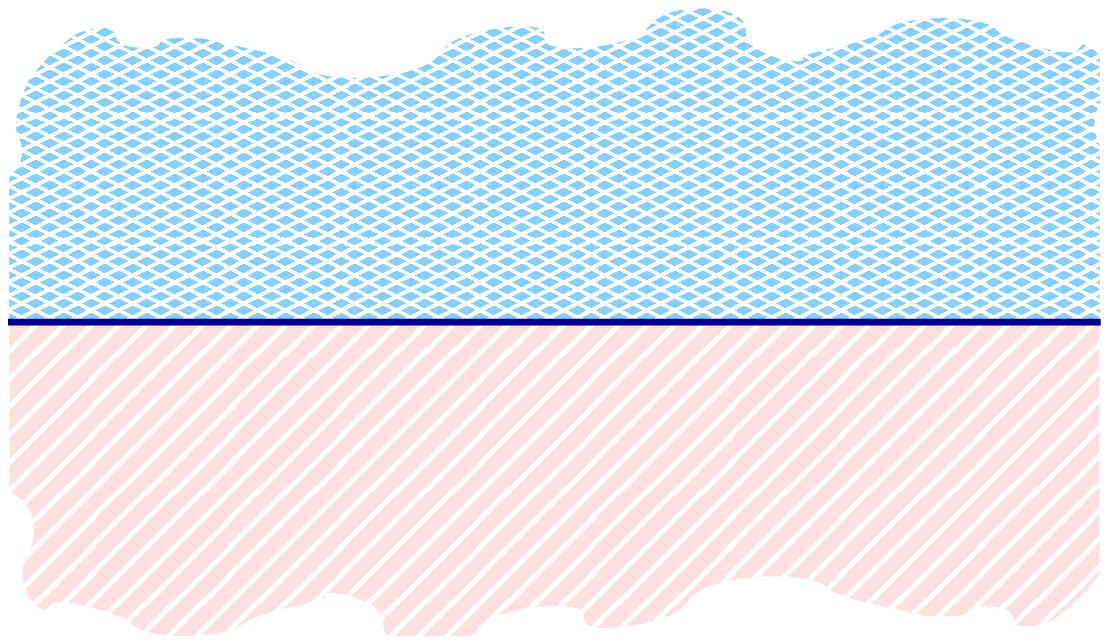}};
  (18,25.8)*{\C_1};
  (43.9,20.5)*{d};
  (23,14.1)*{\C_2};
  \endxy }
  \ee

Again we start with a semisimple fusion category $\calw_d$ of Wilson lines 
that are contained in the defect surface. We refer to such Wilson lines also
as \emph{defect Wilson lines}. In complete analogy with the case of boundary 
conditions, we postulate that there are adiabatic processes of moving
Wilson lines from the bulk on either side of the defect surface 
into the defect surface, whereby they yield defect Wilson lines.
By the same arguments as for boundaries this leads to a central functor
  \be
  \Fto d: \quad \C_1 \to \calw_d
  \ee
and, accounting for relative orientations, to another central functor
  \be
  \Ffrom d: \quad \C_2^\rev \to \calw_d \,.
  \ee

According to Lemma \ref{lem:central}
we thus have two braided functors
  \be
  \tFto d:\quad \C_1 \longrightarrow \Z(\calw_d) 
  \qquad\mbox{and}\qquad
  \tFfrom d:\quad \C_2^\rev\longrightarrow \Z(\calw_d) 
  \ee
as in \eqref{tFto}.
Since $\Z(\calw_d)$ is braided, the images of these two functors commute.
Thus, with the help of the Deligne tensor product, we combine $\tFto d$ 
and $\tFfrom d$ into a single functor
  \be
  \tFtofr d:\quad \C_1\boxtimes \C_2^\rev\longrightarrow \Z(\calw_d) \,.
  \ee
We again invoke a principle of naturality to assert
that the combined functor $\tFtofr d$ is an equivalence of braided categories.

Suppose now that we have a defect Wilson line $W\iN\calw_d$ together with
a rule for exchanging $W$ with any other defect Wilson line $W'\iN\calw_d$. 
The two-dimensional physics of the defect surface does not provide any
natural reason why such a half-braiding rule should exist.
The only possible natural origin of such a rule is that it is related to 
the half-braiding in the three-dimensional parts, using the processes 
encoded in the two functors $\Fto d$ and $\Ffrom d$. 
This amounts to the assumption that the defect Wilson 
line $W$ can be written as a direct sum of fusion products of the form 
$W_1\oti W_2$, where $W_1$ is a defect Wilson line that has been obtained 
by the adiabatic process from $\C_1$, i.e.\ $W_1 \eq \Fto d(L_1)$ for some 
$L_1\iN \C_1$, and similarly $W_2 \eq \Ffrom d(L_2)$ with $L_2\iN \C_2$. 
This shows essential surjectivity of $\tFtofr d$; an argument about 
point-like insertions on Wilson lines that is completely analogous to one 
used for boundary conditions shows that $\tFtofr d$ is fully faithful.

We thus arrive at an equivalence 
  \be
  \C_1^{} \boti \C_2^\rev\stackrel\simeq\longrightarrow \Z(\calw_d)
  \ee
of braided categories that, together with the fusion category $\calw_d$,
is part of the data specifying a surface defect.
We immediately conclude that a topological surface defect joining regions
labeled by the modular tensor categories $\C_1$ and $\C_2$ can only exist
if $\C_1$ and $\C_2$ are in the same Witt class. The existence of such 
an obstruction should not come as a surprise.
Similar effects are, for instance, known from two dimensions: conformal
line defects (and, a fortiori, topological defects) can only exist if the 
two conformal field theories joined by the defect have the same Virasoro 
central charge. In the situation at hand the Witt group -- a concept
that has been introduced for independent reasons,
namely to structure the space of modular tensor categories --
turns out to be the right recipient for the obstruction.

Further, as in the case of boundary conditions, we conclude that other
possible labels of surface defects separating $\C_1$ and $\C_2$ are described 
by module categories over the fusion category $\calw_d$. This also gives
the right bicategorical structure to this collection of surface defects: 
the one of $\calw_d$-modules. By the same type of argument based on 
Proposition \ref{prop:schau13} as in the case of boundaries, it follows that
the other categories of defect Wilson lines come with Witt-trivializations 
of $\C_1^{}\boti \C_2^\rev$ as well,
and that the bicategorical structure does not depend on the choice of $d$.

Again we can consider two special cases to compare our results with existing 
literature. The first -- abelian Chern-Simon theories -- will be relegated to 
section \ref{sec:abCS}. The second is that the TFT on either side of the defect 
surface admits a description of Turaev-Viro type, i.e.\ that both modular 
tensor categories are Drinfeld centers of fusion categories, 
$\C_1\,{\simeq}\,\Z(\cala_1)$ and $\C_2\,{\simeq}\,\Z(\cala_2)$. Using the
identifications
  \be
  \C_1^{}\boxtimes \C_2^\rev \simeq \Z(\cala_1^{})\boxtimes\Z(\cala_2)^\rev
  \simeq \Z(\cala_1\boti\cala_2) \,,
  \ee
where we identify left and right half-braidings, shows that in this case 
the bicategory of $\C_1$-$\C_2$ surface defects can be identified with the
bicategory of $\cala_1$-$\cala_2$-bimodules. Thus in the special case of 
TFTs of Turaev-Viro type our results reproduce those of \cite{kiKon}.

Let us explore some consequences of our results.
First we consider the special case that the surface defect separates two 
regions with the same TFT, i.e.\ that $\C_1 \eq \C_2 \,{=:}\, \C$. By the 
characterization of modular tensor categories given in Definition 
\ref{defi:modular}, there is then a distinguished Witt trivialization,
  \be
  \C\boxtimes \C^\rev \stackrel\simeq \longrightarrow \Z(\calc) \,,
  \ee
which is obtained by using the braiding of the categories \C\ and of $\C^\rev$, 
respectively, to embed them into $\Z(\C)$. This specific surface defect
can be interpreted as a \emph{transparent defect},
very much in the way as a Wilson line labeled by the tensor unit can be 
seen as a transparent Wilson line (and is, for this reason, usually invisible
in a graphical calculus), and accordingly we denote it by the symbol $\tc$.
Indeed, the defect Wilson lines for this 
specific defect are labeled by the objects of \C. The central functor
  \be
  \Fto {\tc}:\quad \calc\to \calc
  \ee
describing a specific
adiabatic process is, as a functor, just the identity. Its
structure of a central functor is then just given by the braiding of \C. In 
physical terms this means that in the adiabatic process labels do not change 
and the braiding is preserved. Similar statements apply to the functor
  \be
  \Ffrom {\tc}:\quad \calc^\rev\to \calc \,,
  \ee
where the structure of a central functor is now given by the opposite
braiding. Thus defect Wilson lines separating the surface defect $\tc $ from 
itself are naturally identified, including the braiding, with ordinary Wilson 
lines in \C. Phrased the other way round: Wilson lines in the three-dimensional
chunk labeled by \C\ can be thought of as being secretly Wilson lines inside 
a defect surface, namely one labeled by the transparent defect $\tc$.

We next discuss implications to surface defects separating \C\ from itself 
of the result of \cite{danO}, reported in Proposition \ref{prop:triple},
that indecomposable \C-modules are in bijection to pairs $(A_1,A_2)$ 
of \'etale algebras in \C\ together with a braided equivalence 
between their full subcategories of local (or dyslectic) modules in \C.
This has the following physical interpretation:
A generic Wilson line in \C\ cannot pass through a given surface defect. 
If, however, a whole package of \C-Wilson lines condenses so as to form 
a local $A_1$-module, the resulting Wilson line can pass through the 
surface defect and reappear on the other side as a condensed package 
of \C-Wilson lines that forms a local $A_2$-module.

This is exactly the type of structure needed in the application of surface 
operators in the TFT construction of the correlators of rational conformal
field theories \cite{fuRs4}, following the suggestions of \cite{kaSau3}. 
The process then puts the fact \cite[Sect.\,4]{mose2}
that the general structure of the bulk partition function is 
``automorphism on top of extension'' in the appropriate and complete setting.
This picture can be easily extended to heterotic theories, for which 
left- and right-moving degrees are in different module tensor categories
$\C_{\rm l}$ and $\C_{\rm r}$. In particular, the obstruction to the existence 
of a heterotic TFT construction based on a pair 
$(\C_{\rm l},\C_{\rm r})$ of modular tensor categories is again captured by 
the Witt group: $\C_{\rm l}$ and $\C_{\rm r}$ must lie in the same Witt class.

The transmission of (bunches of) Wilson lines should be seen as a 
three-dimensional analogue of the following process in two dimensions: 
A topological defect line can wrap around bulk insertions in one full 
conformal field theory to produce a bulk insertion in another theory.
This effect associates a map on bulk fields to any topological defect line. 
This map has, in turn, been instrumental in obtaining classification 
results for defects \cite{fGrs} and in understanding their target space
formulation \cite{fusW}. We expect that the transmission of Wilson lines 
can be used to a similar effect in the situation at hand. That the 
transmission data describe, by Proposition \ref{correspondence}, the 
isomorphism class of a module category, is an encouragement for attempting 
similar classifications as in the two-dimensional case. A first example 
of a classification will be presented in Section \ref{sec:abCS}.

Returning to the case of general pairs $(\C_1,\C_2)$ of modular tensor
categories, the forgetful functor $\varphi_{\!\calw_d}$ from the Drinfeld 
center $\Z(\calw_d)$ to the fusion category $\calw_d$ provides us with
a tensor functor
  \be
  \C_1^{}\boxtimes\C_2^\rev\stackrel\simeq\longrightarrow 
  \Z(\calw_d)
  \stackrel{\varphi_{\!\calw_d}}\longrightarrow\calw_d \,.
  \ee
Via pullback along this functor, the category $\calw_d$ of defect 
Wilson lines comes with a natural structure of a $\C_1$-$\C_2$-bimodule 
category. This bimodule structure arises naturally when one considers
three-valent vertices in the defect surface with two defect Wilson lines
and one bulk Wilson line involved.
This structure should also enter in the description of fusion of 
topological surface defects. We leave a detailed discussion of 
fusion to future work and only remark that the transparent defect 
$\tc$ must act as the identity under fusion.

We conclude with a word of warning: While the structure of a 
$\C_1$-$\C_2$-bimodule on the category $\calw_d$ of defect Wilson lines 
can be expected to have a bearing on fusion, the bicategory
$\C_1$-$\C_2$-\BIMOD\ of $\C_1$-$\C_2$-bimodules can\emph{not} provide 
the proper mathematical model for the bicategory of surface defects.
For instance, taking $\C_1 \eq \C_2 \eq \C$, the natural candidate for 
the transparent defect is \C\ as a bimodule over itself. Using
$\C_1$-$\C_2$-\BIMOD\ as a model for the surface defects, the
Wilson lines separating this transparent defect from itself
would correspond to bimodule endofunctors of \C, and the category  
of these endofunctors is equivalent to $\Z(\C)$ and thus, \C\ 
being modular, to the enveloping category $\C\boti\C^\rev$. We would 
then {\em not} recover ordinary Wilson lines as defect Wilson lines 
in the transparent defect. As we will see in the next section,
taking $\C_1$-$\C_2$-\BIMOD\ as the model for the bicategory of 
surface defects would also contradict results of \cite{kaSau2}
for abelian Chern-Simons theories, and of \cite{kiKon} on 
surface defects in TFTs admitting a description of Turaev-Viro type.


\section{Lagrangian algebras and abelian Chern-Simons theories}\label{sec:abCS}

In this section we describe consequences of our proposal in the special case 
of abelian Chern-Simons theories and compare our findings with the results 
of \cite{kaSau2} for this subclass of TFTs. As a new ingredient our discussion 
involves Lagrangian algebras in the modular tensor category \C\ that labels 
the TFT. Recall from Proposition \ref{correspondence} that Lagrangian algebras 
in the Drinfeld center of a fusion category \cala\ are complete invariants of 
equivalence classes of indecomposable \cala-module categories.

If we are just interested in equivalence classes of indecomposable 
boundary conditions, Lagrangian algebras can be used as follows. The 
presence of a topological boundary condition for a modular tensor 
category \C\ requires the existence of a Witt-trivialization 
$\C \,{\simeq}\, \Z(\cala)$ with $\cala$ a fusion category
providing a reference boundary condition.
Indecomposable or elementary boundary conditions are then in bijection with
indecomposable \cala-modules. The latter, in turn, are in bijection 
with Lagrangian algebras in the Drinfeld center $\Z(\cala)$, which is 
just \C. We can thus classify elementary topological boundary conditions
by classifying Lagrangian algebras in $\C$. This is of considerable 
practical interest, since it acquits us of the task to find an explicit 
Witt-trivialization. However, for many explicit constructions it will 
be important to have the full bicategorical structure at our disposal, 
and this requires an explicit Witt-trivialization. 

The situation for topological surface defects separating modular 
tensor categories $\C_1$ and $\C_2$ is analogous: the classification 
of equivalence classes amounts to classifying Lagrangian algebras in 
$\C_1\boti \C_2^\rev$. Again, this avoids finding a Witt-trivialization,
but does not give direct access to the full bicategorical structure.

\medskip

To make contact to the situation studied in \cite{kaSau2} we first
recall some basic facts about abelian Chern-Simons theories and their 
relation to finite groups with quadratic forms.
Let $\Lambda$ be a a free abelian group of rank $n$ and
$V \,{:=}\,\Lambda \,{\otimes_{\mathbb Z}}\, \mathbb{R}$ the
corresponding real vector space. Denote by $\mathbb{T}_{\Lambda}$ the 
torus $V/\Lambda$. The classical abelian Chern-Simons theory with structure 
group $\mathbb{T}_{\Lambda}$ is completely determined by the choice of a 
symmetric bilinear form $K$ on $V$ whose restriction to the additive subgroup
$\Lambda$ is integer-valued and even. We call the pair $(\Lambda,K)$ an
even \emph{lattice} of rank $n$.

\begin{defi} \mbox{}\\[1pt]
(i)\, A \emph{bicharacter}, with values in $\complexx$, on a finite abelian 
group $D$ is a bimultiplicative map $\beta\colon D\,{\times} D\To\complexx$. 
\\
A \emph{symmetric} bicharacter, or symmetric bilinear form, on $D$
is a bicharacter $\beta$ satisfying $\beta(x, y) \eq \beta(y, x)$ for all
$x,y\iN D$. 
\\[3pt]
(ii)\, A \emph{quadratic form} on a finite abelian group $D$ is a function 
$q \colon D\To \complexx$ such that $q(x) \eq q(x^{-1})$ and such that 
$\beta(x,y) \,{:=}\, q(x{\cdot}y)/q(x)\, q(y)$ is a symmetric bilinear form. A
\emph{quadratic group} is a finite abelian group endowed with a quadratic form.
\end{defi}

To the lattice $(\Lambda,K)$ we associate a finite group with a quadratic 
form in the following way. Set $\Lambda^{*} \,{:=}\, 
\Hom_{\mathbb Z}(\Lambda,\mathbb{R})$, and denote by $\Im\,K$ the image of 
$\Lambda$ in $\Lambda^{*}$ under the canonical map 
$K\colon \Lambda\to{\Lambda^{*}}$. The finite abelian group 
$D \,{:=}\, \Lambda^{*}/\Im\,K$ is called the \emph{discriminant group} of 
the lattice $(\Lambda,K)$. Since the symmetric bilinear form $K$ is integer 
and even, the group $D$ comes equipped with a quadratic form $q \colon D 
\To \mathbb{C}^\times$, with $q(\mu) \eq \exp(2\pi\ii\, K(\mu,\mu))$. 

Different lattices may give rise to discriminant groups that are 
isomorphic as quadratic groups. As argued in \cite{beMo}, many 
properties of quantum Chern-Simons theory are encoded in the pair $(D,q)$. 
The pair $(D,q)$ determines, in turn, an equivalence class of 
braided monoidal categories, which we denote by $\C(D,q)$.
For completeness we briefly give some details on how this category is 
constructed (for more details see \cite{joSt6,kaSau2}).
First, recall that for any abelian group $A$ there is a bijection
  \be
  H_{\rm{ab}}^{3}(A;\complexx)
  \stackrel\simeq\longrightarrow \textrm{Quad}(A) \,,
  \label{eq:abelian}
  \ee
between the group $\textrm{Quad}(A)$ of quadratic forms on $A$ taking 
values in $\complexx$ and the third abelian cohomology group 
$H_{\rm{ab}}^{3}(A;\complexx)$ of $A$ \cite[Thm.\,3]{macl2}. 
A representative for a  class  in $H_{\rm{ab}}^{3}(A;\complexx)$ 
is given by a pair $(\psi,\Omega)$ consisting of a 3-cocyle $\psi$ in the 
ordinary group cohomology of $A$ and a 2-cochain $\Omega$, satisfying some 
compatibility conditions (which imply the validity of the hexagon axioms for 
the braiding \eqref{cxy} below).  Given a pair $(\psi,\Omega)$ representing
an abelian 3-cocycle, we obtain a quadratic form $q$ on $A$ by setting 
$q(a) \,{:=}\, \Omega(a,a)$ for $a\iN{A}$. This realizes one direction of
the isomorphism (\ref{eq:abelian}). 
 
On the other hand, given a quadratic form $q$
on $A$, we obtain a pre-image  $(\psi,\Omega)$ only upon additional 
choices; one possible choice is an ordered set of generators of 
the abelian group $A$. We will ignore this subtlety in the following
and omit the label $(\psi,\Omega)$ from the notation.

Consider now a quadratic group $(D,q)$ and  choose an abelian 3-cocycle
$(\psi,\Omega)$ representing the quadratic form $q$ in 
$H_{\rm{ab}}^{3}(D;\complexx)$. As an abelian category,
$\C(D,q)$ is the category of finite-dimensional complex $D$-graded 
vector spaces and graded linear maps. The simple objects of this category are 
complex lines $\complex_x$ labeled by group elements $x\iN D$. In particular 
we have $\Hom(\complex_x,\complex_x)\,{\cong}\,\complex$. We equip the category 
$\C(D,q)$ with the tensor product of complex vector spaces, but with 
associator given by the 3-cocyle $\psi$. The 2-cochain $\Omega$ induces a 
braiding $c$ on this monoidal category; the braiding acts on simple objects as
  \be
  \begin{array}{rl}
  c_{xy}:\quad \complex_x \otimes \complex_y &\!\!\stackrel\simeq\longrightarrow\,
  \complex_y \otimes \complex_x \\[4pt]
  v\oti w &\!\! \longmapsto\, \Omega(x,y)\, w \oti v \,.
  \end{array}
  \label{cxy}\ee
The braided pointed fusion category
thus obtained depends, up to equivalence of braided monoidal categories, 
only on the class $[(\psi,\Omega)]$ in abelian cohomology\cite{joSt6}.

Taking the reverse category amounts to replacing the quadratic form $q$ 
by the quadratic form $q^{-1}$ which takes inverse values, i.e.\ 
${(\C(D,q)}^\rev {\cong}\, \C(D,q^{-1})$, while the Deligne product amounts 
to taking the direct sums of the groups and of the quadratic forms. 
In other words, one has 

\begin{lemma}
Let $(D_{1},q_{1})$ and $(D_{2},q_{2})$ be finite groups with quadratic 
forms. Then
  \be\label{D1boxD2=D1+D2}
  \C(D_{1},q_{1})\boxtimes {\C(D_{2},q_{2})}^\rev
  \simeq\, \C(D_1{\oplus}{D_2},q_{1}{\oplus}q^{-1}_{2})
  \ee
as braided monoidal categories.
\end{lemma} 

A quadratic form $(D,q)$ is said to be \emph{non-degenerate} iff the 
associated symmetric bilinear form is non-degenerate in the sense that 
the associated group homomorphism $D\To \Hom(D,\complexx)$ is an 
isomorphism. A basic fact about categories of the type $\C(D,q)$ is 

\begin{lemma}{\rm \cite[Sec.\,5.3.]{dmno}}\\
The braided monoidal category $\C(D,q)$ is modular iff 
the quadratic form $q$ is non-degenerate.
\end{lemma} 

\medskip

In the present context the role of the modular tensor category $\C(D,q)$ 
is as the category of (bulk) Wilson lines in the Chern-Simons theory 
corresponding to a lattice with discriminant group $(D,q)$. We now 
make our proposal for boundary 
conditions and surface defects explicit in this case. To this end 
we need an explicit description of Lagrangian algebras.

\begin{defi}\cite[Sect.\,2.4]{enoM}\\[2pt]
(i)\, A \emph{metric group} is a quadratic group $(D;q)$ for which
the quadratic form $q$ is non-degenerate.
\\[4pt]
(ii)\, For $U$ a subgroup of a quadratic group $(D,q)$ with symmetric
bilinear form $\beta\colon D\,{\times}\, D\To\complexx$, 
the \emph{orthogonal complement} $U^\perp$ of $U$ is the set of all 
$d\iN D$ such that $\beta(d,u) \eq 1$ for all $u\iN U$. 
\\[2pt]
If $q$ is non-degenerate, $U^\perp$ is isomorphic to $D/U$,
so that $|D| \eq |U|\cdot |U^\perp|$.
\\[4pt]
(iii)\, Let $(D,q)$ be a metric group. A subgroup $U$ of $D$ is said to be 
\emph{isotropic} iff $q(u) \eq 1$ for all $u\iN U$. 
\\[4pt]
(iv)\, For any isotropic subgroup $U$ of a metric group $(D,q)$ there exists
an injection $U\,{\hookrightarrow}\, (D/U)^*$, so that $|U|^2\,{\le}\, |D|$.
An isotropic subgroup $L$ of $D$ is called \emph{Lagrangian} iff $|L|^2 \eq |D|$.
\end{defi}

The concept of a Lagrangian subgroup is linked to Lagrangian algebras by the 
following assertion, which is a corollary of the results in 
\cite[Sect.\,2.8]{dgno2}.

\begin{thm} 
Let $D$ be a finite abelian group with a nondegenerate quadratic form $q$. 
There is a bijection between Lagrangian subgroups of $D$ and Lagrangian 
algebras in \,$\C(D,q)$.
\end{thm}    

\medskip

We thus arrive at the following two statements:
  \def\leftmargini{1.7em}~\\[-1.55em]\begin{itemize}\addtolength{\itemsep}{-7pt}
  \item[(1)]
Elementary topological boundary conditions for the abelian Chern-Simons 
theory based on the modular tensor category $\C(D;q)$ are in bijection 
with Lagrangian algebras in $\C(D;q)$ and thus 
with Lagrangian subgroups of the metric group $(D,q)$.
  \item[(2)]
Elementary topological surface defects separating the abelian Chern-Simons 
theories based on the modular tensor categories $\C(D_1;q_1)$ and 
$\C(D_2;q_2)$ are in bijection with Lagrangian subgroups of the metric
group $(D_1{\oplus}D_2,q_1{\oplus}q_2^{-1})$.
\end{itemize}

The first of these results was established in \cite{kaSau2} by an explicit 
analysis using Lagrangian field theory. The second result was then deduced 
from the first by arguments based on the folding trick.

As a particular case, consider the transparent surface defect $\tc$ 
separating $\C(D,q)$ from itself. It corresponds to the canonical 
trivialization $\C(D,q)\boti {\C(D,q)}^\rev\,{\simeq}\, \Z(\C(D,q))$.
The Lagrangian algebra corresponding to $\tc$ is given by
the Cardy algebra $\bigoplus_{X\in \Irr(\C)}X\boti X^\vee_{}$
and corresponds to the diagonal subgroup in $D\oplus{D}$, in 
accordance with the results in \cite[Section 3.3]{kaSau2}.


\section{Relation with special symmetric Frobenius algebras}\label{sec:Frob}

In this section we explain how in our framework special symmetric
Frobenius algebras can be obtained from certain surface defects $S$ separating
a modular tensor category $\C$ from itself and a Wilson line separating
$S$ from the transparent defect $\tc$. Our results provide a rigorous
mathematical foundation for the ideas of \cite{kaSau3}. A central
tool in our study are string diagrams.

\subsection{String diagrams}

A \emph{string diagram} is a planar diagram describing morphisms in a 
bicategory. Such diagrams are Poincar\'e dual to another type of diagram 
frequently used for bicategories, in which 2-morphisms are 
attached to 2-dimensional parts of the diagram. String diagrams 
are particularly convenient for encoding properties 
of adjointness and biadjointness in a graphical calculus. 
String diagrams apply in particular to the bicategory of small categories
in which 1-morphisms are functors and 2-morphisms are natural
transformations, an example the reader might wish to keep in mind. 
For more details see e.g.\ \cite{laud,khov9}.

We fix a bicategory; in a first step, we only consider objects and 
1-morphisms. They can be visualized in one-dimensional diagrams, with 
one-dimensional segments describing objects and zero-dimensional parts 
indicating 1-morphisms. In our convention, such diagrams are drawn 
horizontally and are to be read from right to left. Thus for \cala\ and 
\calb\ objects of the bicategory and a 1-morphism $F\colon \cala\To\calb$, 
we draw the diagram
  \be
  \parbox{100pt}{ \xy
  (20.2,23.4)*{F};
  (6,23)*{\mathcal{B}};
  (35,23)*{\mathcal{A}};
  (5,16)*{\phantom.};
  (5,26)*{\phantom.};
  (20,20)*{\includegraphics[scale=0.3]{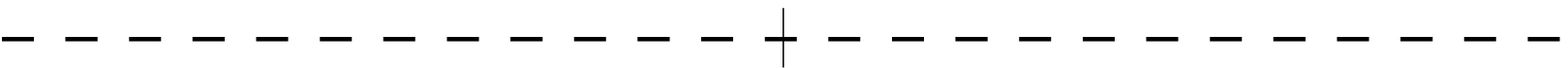}};
  \endxy }
  \ee
The composition $F_n\,{\cdots}\,F_1 \,{\equiv}\, F_n\cir{\cdots}\cir F_1 \colon 
\cala_{1}\To\mathcal{A}_{n}$ of 1-morphisms $F_{i}\colon \cala_{i}\To
\mathcal{A}_{i+1}$ is represented by horizontal concatenation
  \be
  \parbox{100pt}{ \xy
  (37.3,24)*{F_{2}};
  (60,24)*{F_{1}};
  (48,23)*{\mathcal{A}_{2}};
  (28,23)*{\mathcal{A}_{3}};
  (70,23)*{\mathcal{A}_{1}};
  (-15,24)*{F_{n}};
  (-5,23)*{\mathcal{A}_{n}};
  (-25,23)*{\mathcal{A}_{n+1}};
  (5,16)*{\phantom.};
  (5,27)*{\phantom.};
  (22.5,20)*{\includegraphics[scale=0.6]{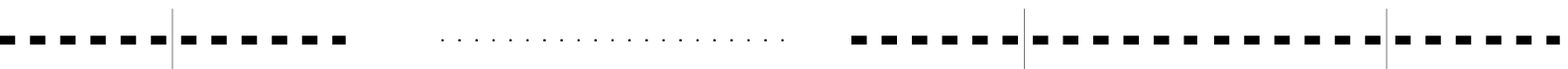}};
  \endxy } 
  \ee

To accommodate also 2-morphisms a second dimension is needed.
Objects are now represented by two-dimensional regions and 1-morphisms by
one-dimensional vertical segments, while zero-dimensional parts indicate
2-morphisms. In our convention, the vertical direction is to be read from
bottom to top. Thus a 2-morphism $\alpha\colon F_{1}\,{\Rightarrow}\,{F_{2}}$ 
between 1-morphisms $F_1,F_2$ from objects $\mathcal{A}$ to $\mathcal{B}$ is 
depicted by the diagram
  \be
  \parbox{100pt}{ \xy
  (25,3.8)*{F_{1}};
  (5,.5)*{\phantom.};
  (25.3,36)*{F_{2}};
  (22,20.5)*{\alpha};
  (11,23)*{\mathcal{B}};
  (38,23)*{\mathcal{A}};
  (25,20)*{\includegraphics[scale=0.4]{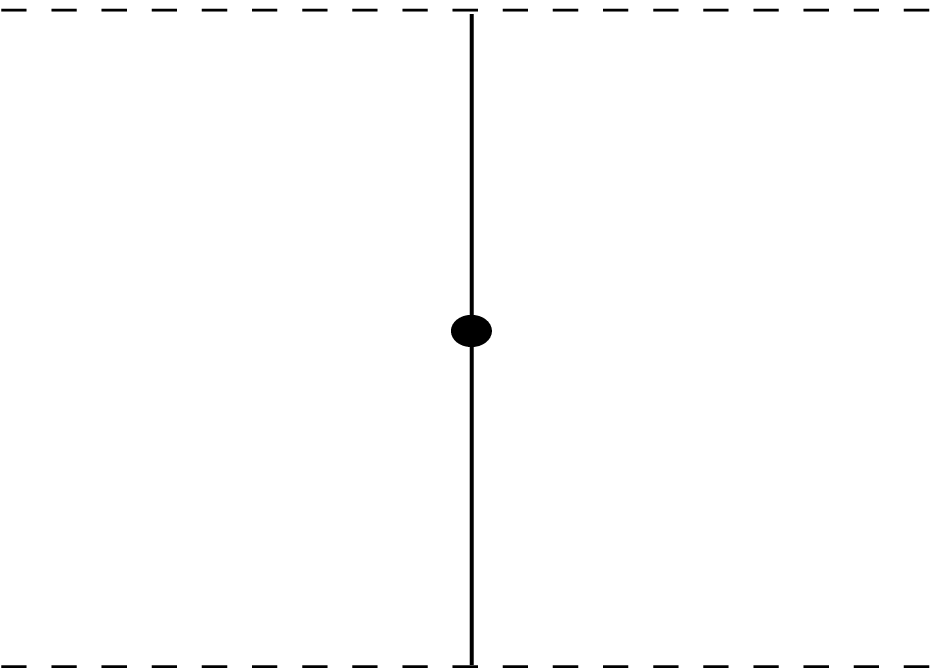}};
  \endxy } 
  \ee
For the moment, we require that the strands always go from bottom to top
and do not allow `U-turns'.
For the identity 2-morphism $\alpha \eq \id_F$ we omit the blob 
in the diagram. For the identity 1-morphism $\Id_\cala$ we omit any label
except for the one referring to the object $\mathcal{A}$. With
these conventions, 2-morphisms $\alpha\colon F\,{\Rightarrow}
\Id_{\mathcal{A}}$ and $\beta\colon \Id_{\mathcal{A}}\,{\Rightarrow}\,{F}$ 
with $F$ an endo-1-morphism of the object $\mathcal{A}$ are drawn as
  \be
  \parbox{100pt}{ \xy 
  (5,2.1)*{\phantom.};
  (21.2,4)*{F};
  (20,17.9)*{\alpha};
  (30,24)*{\mathcal A};
  (22,20)*{\includegraphics[scale=0.4]{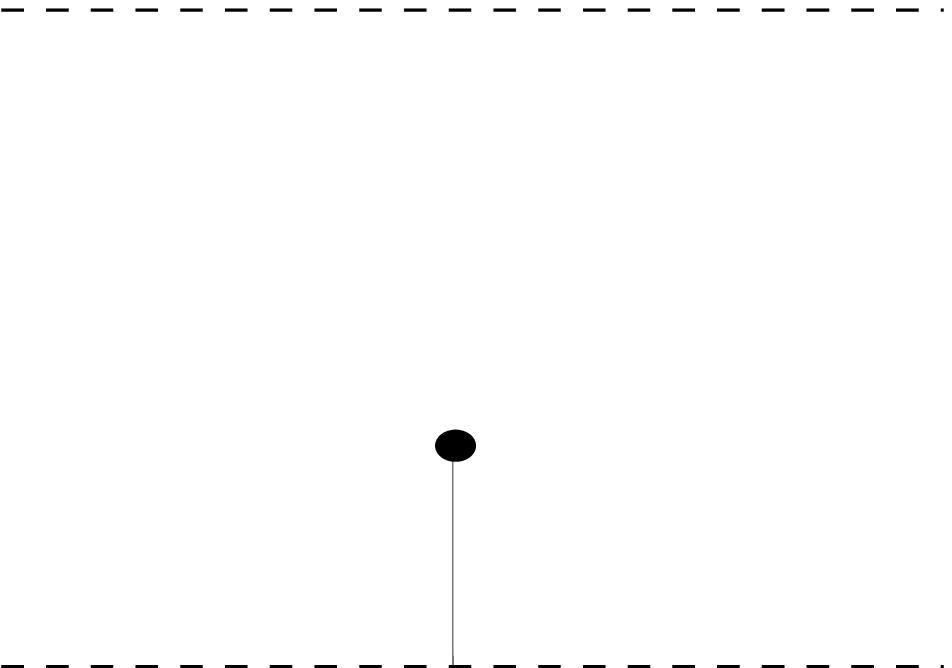}};
  (53,20)*{\text{and}};
  (86.7,36)*{F};
  (88.9,23.5)*{\beta};
  (97,18)*{\mathcal A};
  (87,28.1)*{\includegraphics[scale=0.4]{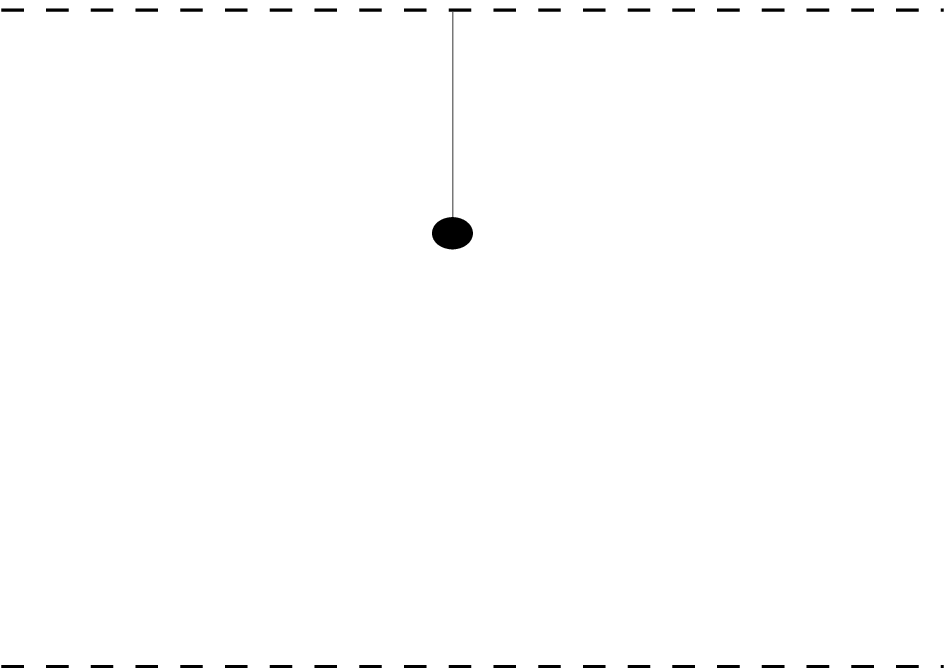}};
  \endxy}  
  \ee
respectively, while a
natural transformation $F_{2}F_{1}\,{\Rightarrow}\,\Id_\cala$ is represented by
  \be
  \parbox{100pt}{ \xy
  (17,3.8)*{F_{2}};
  (31.7,3.8)*{F_{1}};
  (5,1.5)*{\phantom.};
  (5,35.5)*{\phantom.};
  (23,18)*{\alpha};  
  (25,10.8)*{\mathcal B};
  (33,23)*{\mathcal A};
  (25,20)*{\includegraphics[scale=0.4]{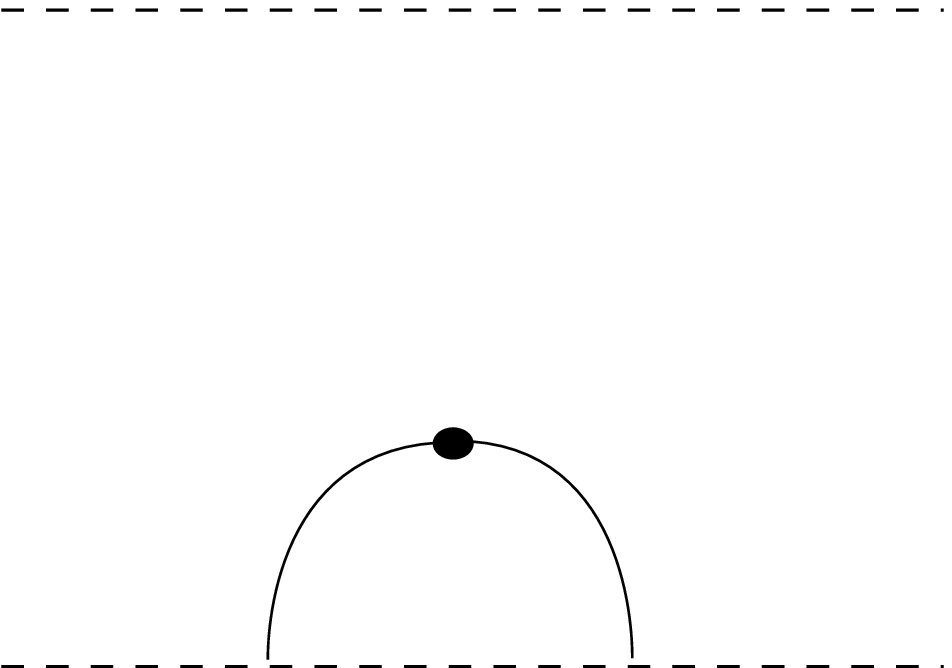}};
  \endxy } 
  \ee

The 2-morphisms of a bicategory can be composed horizontally and
vertically. Horizontal composition is depicted as juxtaposition, as in
  \be
  \parbox{100pt}{ \xy 
  (25,20)*{\includegraphics[scale=0.4]{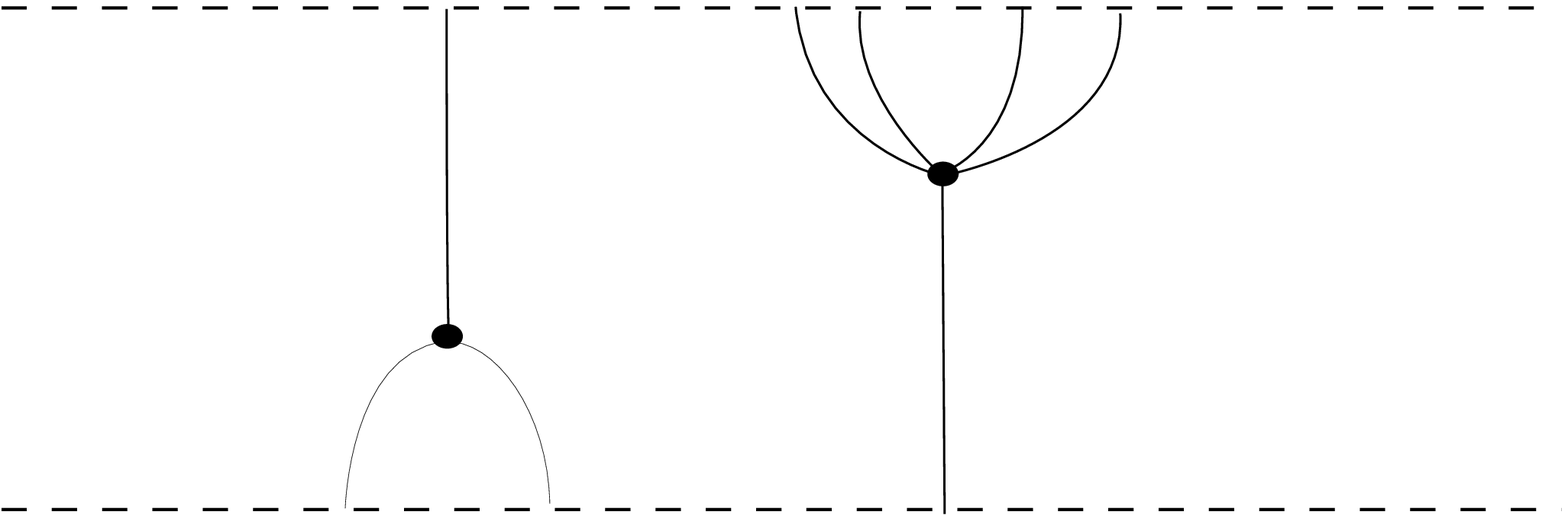}};
  (5,1.5)*{\phantom.};
  (5,35.5)*{\phantom.};
  (-31,20)*{\alpha\otimes\beta~=};
  (10,17.5)*{\alpha};
  (31.5,22.6)*{\beta};
  \endxy }  
  \ee
Vertical composition is represented as vertical concatenation of diagrams;
thus e.g.
  \be
  \parbox{100pt}{ \xy 
  (25,20)*{\includegraphics[scale=0.4]{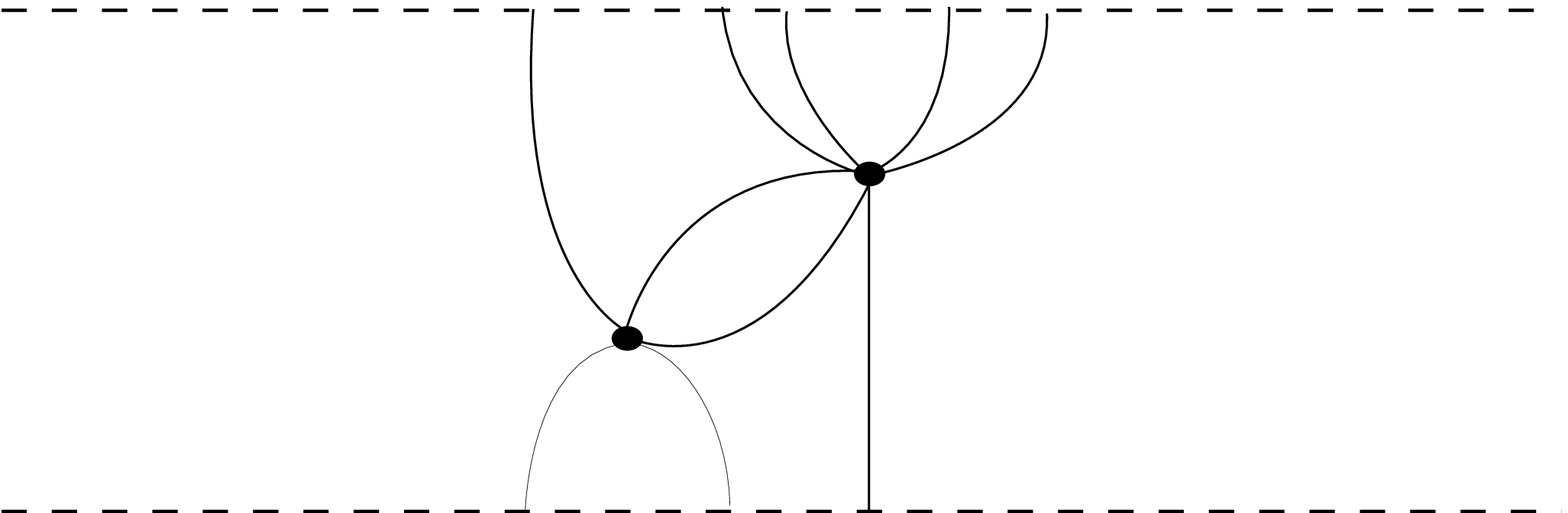}};
  (-47,20)*{(\id_{G}\oti\beta) \circ (\alpha\oti\id_{F}) ~=};
  (14,16.6)*{\alpha};
  (31.5,22.8)*{\beta};
  (12.1,35.7)*{G};
  (29.5,3.9)*{F};
  \endxy }  
  \ee

\medskip

In the bicategory of small categories, we have the notion of an
adjoint functor. This notion can be generalized to any 1-morphism
in a bicategory.
Given two 1-morphisms $F\colon \mathcal{A}\To{\mathcal{B}}$ and
$G\colon \mathcal{B}\To{\mathcal{A}}$, $G$ is said to be 
\emph{right adjoint} to $F$, and $F$ \emph{left adjoint} to $G$, iff
there exist 2-morphisms
  \be
  \eta:\quad \Id_{\mathcal{A}}\Rightarrow{GF} \qquand
  \eps:\quad {FG}\Rightarrow\textrm{Id}_{\mathcal{B}}
  \ee
satisfying
  \be
  (\id_{F}\oti\eta)\circ(\eps\oti\id_{F}) = \id_{F} \qquand
  (\eta\oti\id_{G})\circ(\id_{G}\oti\eps) = \id_{G} \,.
  \label{adj}\ee
The 2-morphisms $\eta$ and $\eps$, if they exist, are not unique.
For any number $\lambda\iN\complexx$ we can replace $\eta$ by 
$\lambda\,\eta$ and $\eps$ by $\lambda^{-1} \eps$ to get another pair of
morphisms. For each such pair, $\eta$ is called a \emph{unit} and 
$\eps$ a \emph{counit} of the \emph{adjoint pair} $(F,G)$.

In the diagrammatic description, special notation is introduced for the 
unit and counit of an adjoint pair of functors: we depict them as
  \be
  \parbox{100pt}{ \xy 
  (3,29.1)*{\includegraphics[scale=0.4]{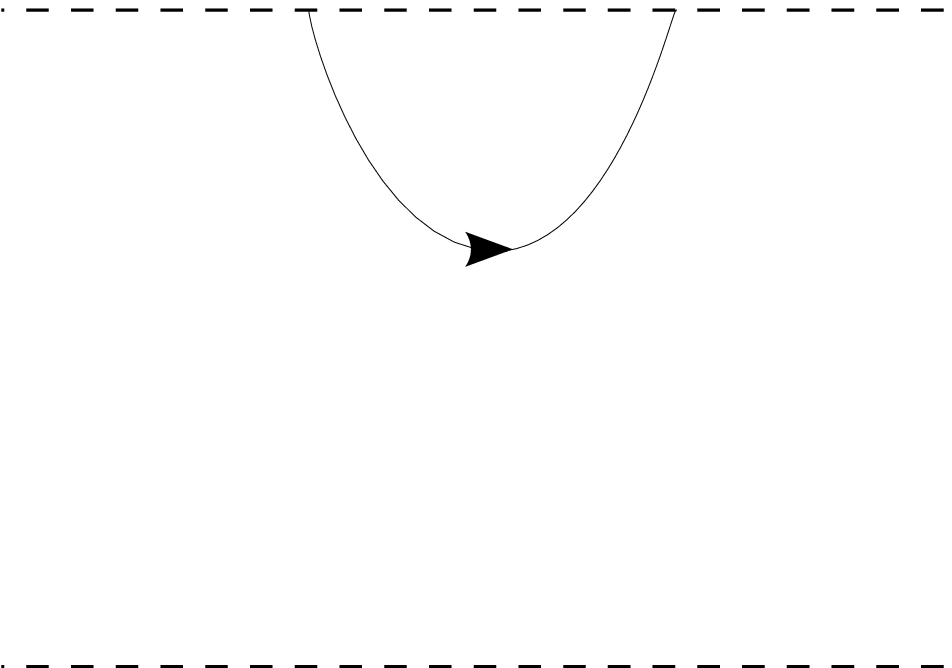}};
  (-25.3,19.9)*{\eta~=};
  (12,17)*{\mathcal{A}};
  (4.2,30)*{\mathcal{B}};
  (-3,37)*{G};
  (11.9,37)*{F};
  (40.3,20.4)*{\text{and}};
  (56.5,19.9)*{\eps~=};
  (87,20.8)*{\includegraphics[scale=0.4]{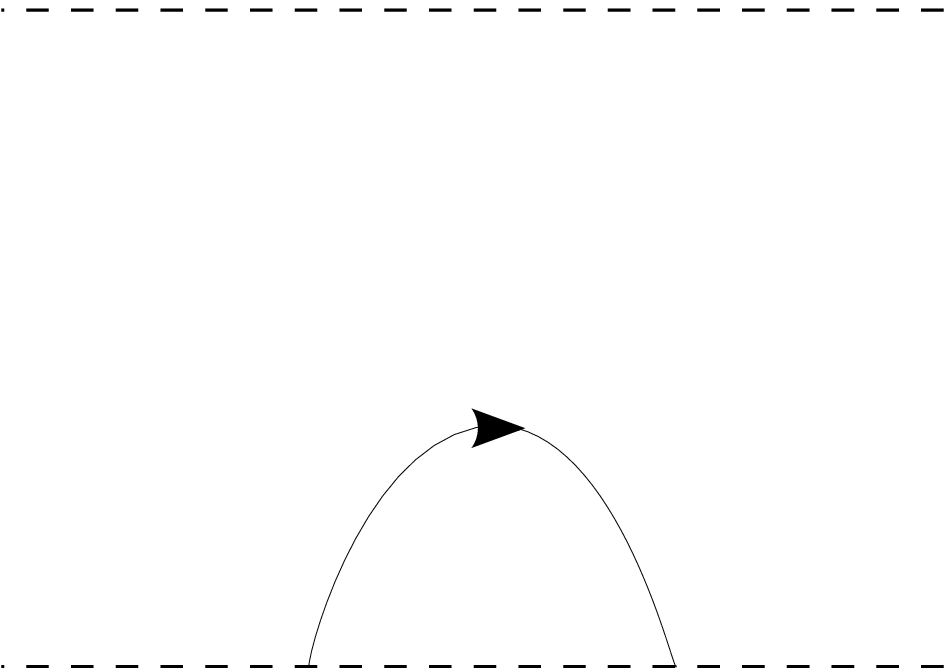}};
  (95,26)*{\mathcal{B}};
  (87.5,11.8)*{\mathcal{A}};
  (80,4.4)*{F};
  (95.6,4.4)*{G};
  \endxy}  
  \ee
The equalities \eqref{adj} amount to the identifications 
  \be
  \label{zickezacke}
  \parbox{100pt}{ \xy 
  (0,60)*{\includegraphics[scale=0.4]{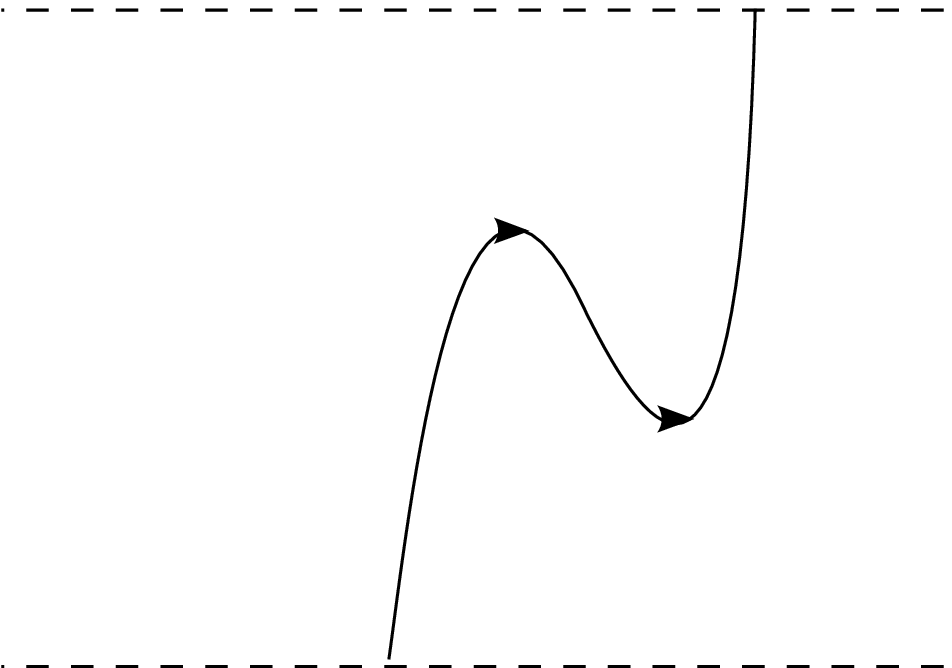}};
  (65,60)*{\includegraphics[scale=0.4]{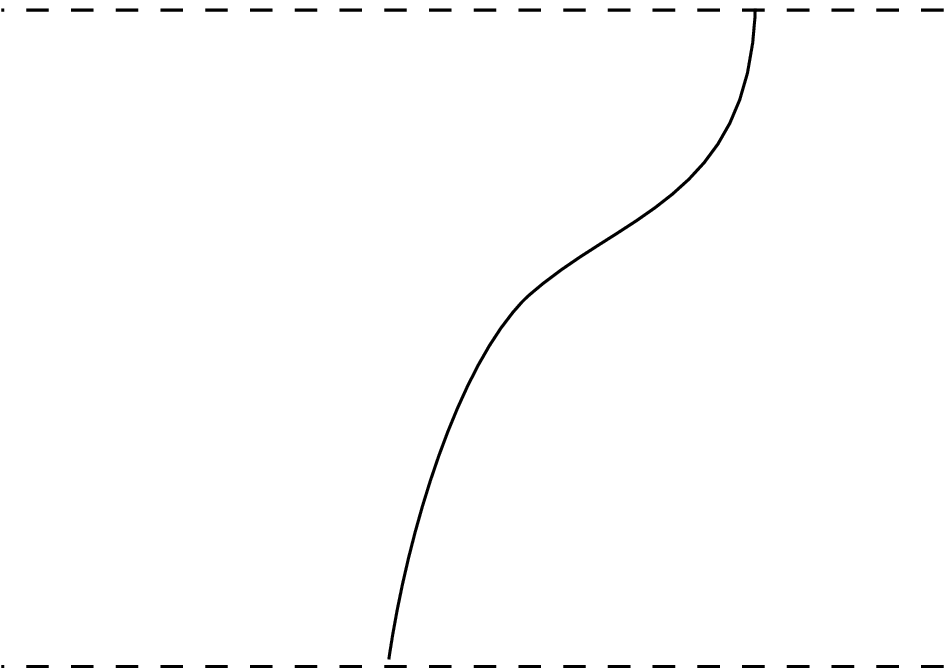}};
  (0,20)*{\includegraphics[scale=0.4]{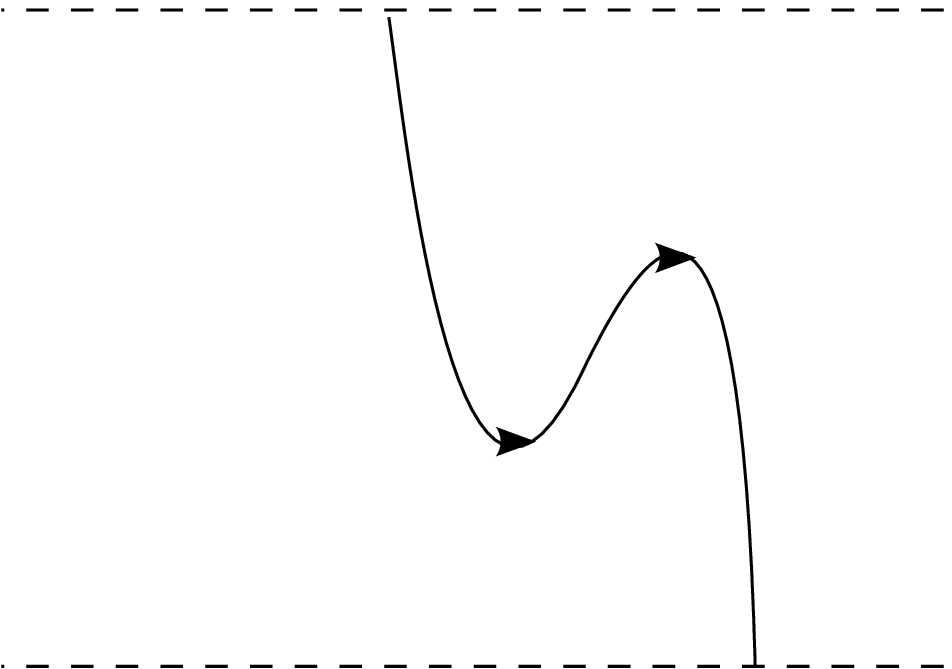}};
  (65,20)*{\includegraphics[scale=0.4]{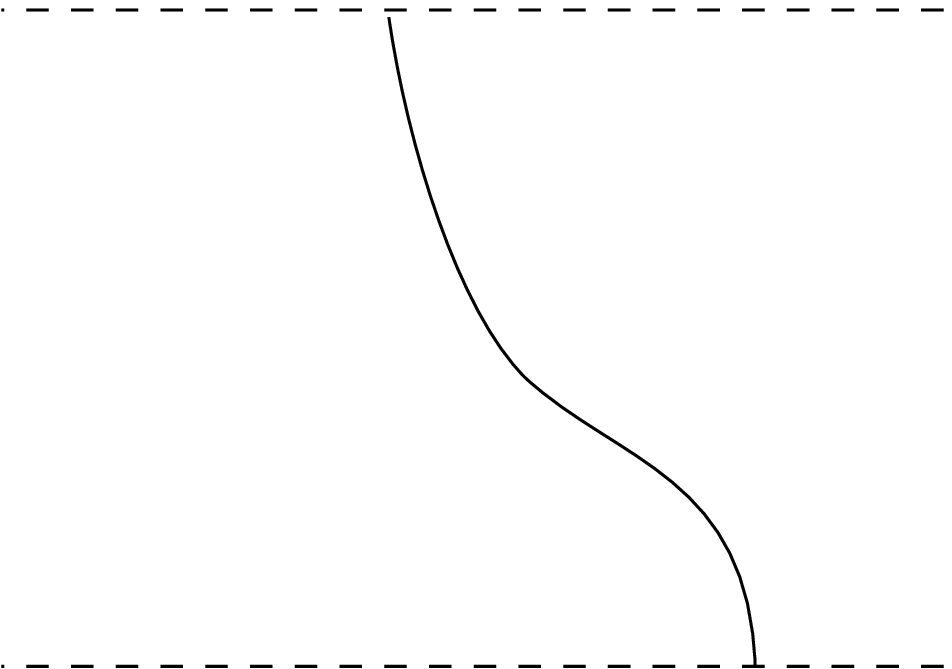}};
  (11.4,75.7)*{F};
  (76.5,75.7)*{F};
  (-4.0,43.5)*{F};
  (61.6,43.5)*{F};
  (33,20)*{=};
  (33,60)*{=};
  (-2.9,35.8)*{G};
  (62.3,35.8)*{G};
  (11.1,3.8)*{G};
  (76.4,3.8)*{G};
  \endxy   }  
  \ee
of diagrams.

In general, the existence of a left adjoint functor does not
imply the existence of a right adjoint functor. Even if both
adjoints exist, they need not coincide. The same statements
hold for left and right adjoints of a 1-morphism in an
arbitrary bicategory. It therefore makes sense to give the

\begin{defi} 
A 1-morphism $F$ in a bicategory is called \emph{biadjoint} to a 
1-morphism $G$ iff it is both a left and a right adjoint of $G$.
Since then $G$ is both left and right adjoint to $F$ as well,
such a pair $(F,G)$ of 1-morphisms is called a \emph{biadjoint pair}.
The adjunction $(F,G)$ is then called \emph{ambidextrous}.
\end{defi}

For a biadjoint pair $(F,G)$, we thus have, apart from the 2-morphisms
$\eta$ and $\eps$ introduced in formula \eqref{adj}, additional 2-morphisms
  \be
  \tilde\eta:\quad \Id_{\mathcal{B}} \,\Rightarrow{FG} \qquand
  \tilde\eps:\quad {GF} \,\Rightarrow\textrm{Id}_{\mathcal{A}}
  \ee
satisfying zigzag identities analogous to
the identities \eqref{zickezacke} for $\eta$ and $\eps$.

Once we restrict ourselves to string diagrams in which all lines are 
labeled by 1-morphisms admitting an ambidextruous adjoint, and having 
fixed adjunction 2-morphisms,
we can allow for lines with U-turns in string diagrams with the
appropriate one of the four adjunction 2-morphisms at the cups and
caps, because the relations we have just presented allow us to consistently 
apply isotopies to all lines. We thus obtain complete isotopy invariance.

For a biadjoint pair $(F,G)$ one can in particular form the composition
$\tilde\eps \cir \eta$, which is an endomorphism of the identity functor 
$\Id_\cala$, as well as $\eps \cir \tilde\eta$ which is an endomorphism of 
$\Id_\calb$. Graphically,
  \be
  \parbox{100pt}{ \xy 
  (25,20)*{\includegraphics[scale=0.4]{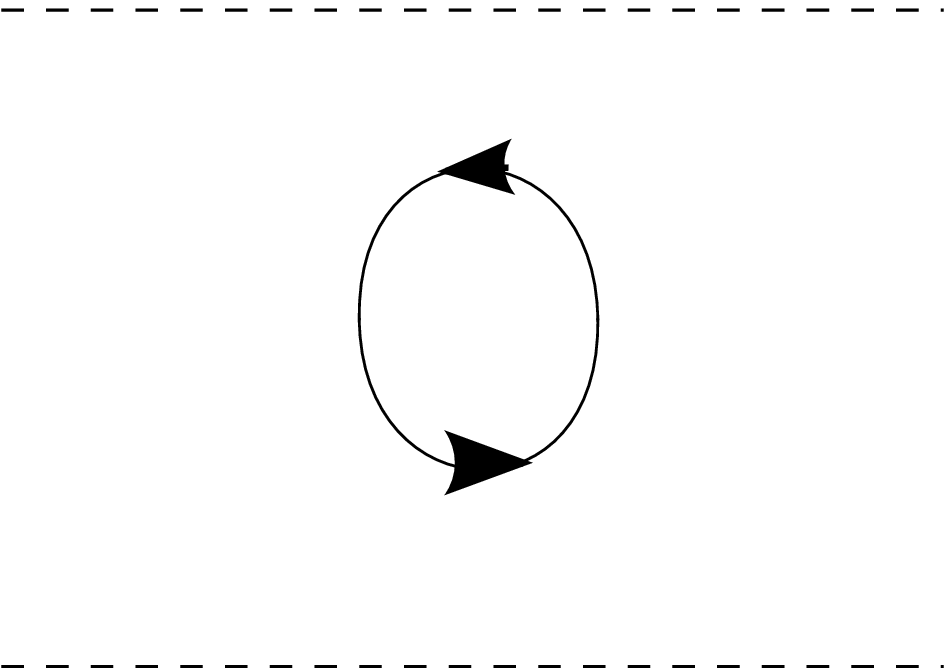}};
  (-8,16)*{\tilde\eps \circ \eta ~=};
  (11,22)*{\mathcal A};
  (25.1,17.2)*{\mathcal B};
  (0,3)*{\quad};
  (100,20)*{\includegraphics[scale=0.4]{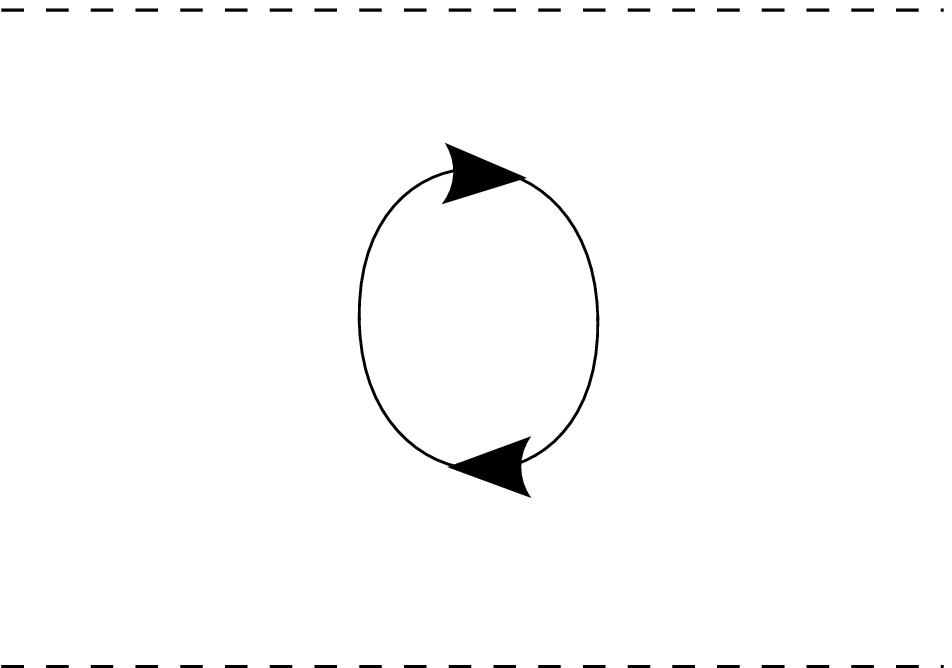}};
  (67,16)*{\eps \circ \tilde\eta ~=};
  (86,22)*{\mathcal B};
  (100.1,17.2)*{\mathcal A};
  \endxy }  
  \ee
It should be appreciated that if the adjunction 2-morphisms are
rescaled, the endomorphisms of $\Id_\cala$ and $\Id_\calb$ 
which appear here get rescaled by reciprocal factors.


\subsection{Frobenius algebras from string diagrams}\label{sec:Fstr}

So far our discussion concerned general bicategories. We now turn
to the bicategory of surface operators separating modular tensor
categories. As was argued in \cite{kaSau3}, from
a surface defect $S$ separating a modular tensor category \C\
from itself, we expect to be able to construct a symmetric special 
Frobenius algebra in \C\ for each Wilson line separating $S$ and
the transparent defect $\tc$. Different Wilson lines should yield
Morita equivalent Frobenius algebras. We now give a
proof of this fact that is based on our description of surface defects.

We thus consider a $\C$-module $S$. Recall that a Wilson line $M \iN 
\HOM(S,\tc)$ is described by a $\C$-module functor $M\colon S\To \cc$.

\begin{lemma}
Let \C\ be a modular tensor category, $S$ an object in \C-\MOD\ and $M \iN 
\HOM_\C(S,\tc)$. Then the functor $M$ has a biadjoint as a module functor.
\end{lemma}

\begin{proof} 
$M$ is an additive functor between semisimple \complex-linear categories.
Now as an abelian category, a finitely semisimple \complex-linear category 
is equivalent to $(\Vectc)^{\boxtimes n}$ with $n \eq |\Irr(\C)|$ the (finite)
number of isomorphism classes of simple objects.
Moreover, any additive endofunctor $F$ of \Vectc\ is given by tensoring with 
the vector space $V \eq F(\complex)$, and it is ambidextrous, the adjoint 
being given by tensoring with the dual vector space $V^*_{}$. It follows that 
the functor $M$ is equivalent to a functor $\tilde{M}\colon 
(\Vectc)^{\boxtimes n} \To (\Vectc)^{\boxtimes m}$ for some integers $n$ and
$m$ and is completely specified by an $n\Times{m}$-matrix of \complex-vector 
spaces. Further, both the left and the right adjoint functor to 
$\tilde{M}$ are then given by the `adjoint' matrix, and hence $\tilde{M}$ is
ambidextrous. As a consequence, $M$ is ambidextrous as a functor. Using 
arguments from \cite{etno}, the bi-adjoint of $M$ has two structures of a 
module functor, from being a left adjoint and right adjoint, respectively. 
These two structures coincide.
\end{proof} 

Since all adjunctions involved are ambidextruous, we can from now on freely
use isotopies in the manipulations of string diagrams. We next consider
the following construction for any module functor $M \iN \HOM_\C(S,\tc)$.
Denote by $\bar M$ the module functor biadjoint to $M$ and set 
$A^M \,{:=}\, M\cir\bar M$. Then $A^{M}\iN\HOM(\tc,\tc)$, which by 
Proposition \ref{gibtsdoch} is equivalent to \C\ as a monoidal category. We 
proceed to equip the object $A^M\iN\C$ with the structure of a Frobenius 
algebra in \C.  For the product, we introduce the morphism 
$m_{A^M_{}}\colon A^{M}\oti{A^{M}}\To{A^{M}}$ as
  \be
  \parbox{100pt}{ \xy
  (19,19)*{\includegraphics[scale=0.4]{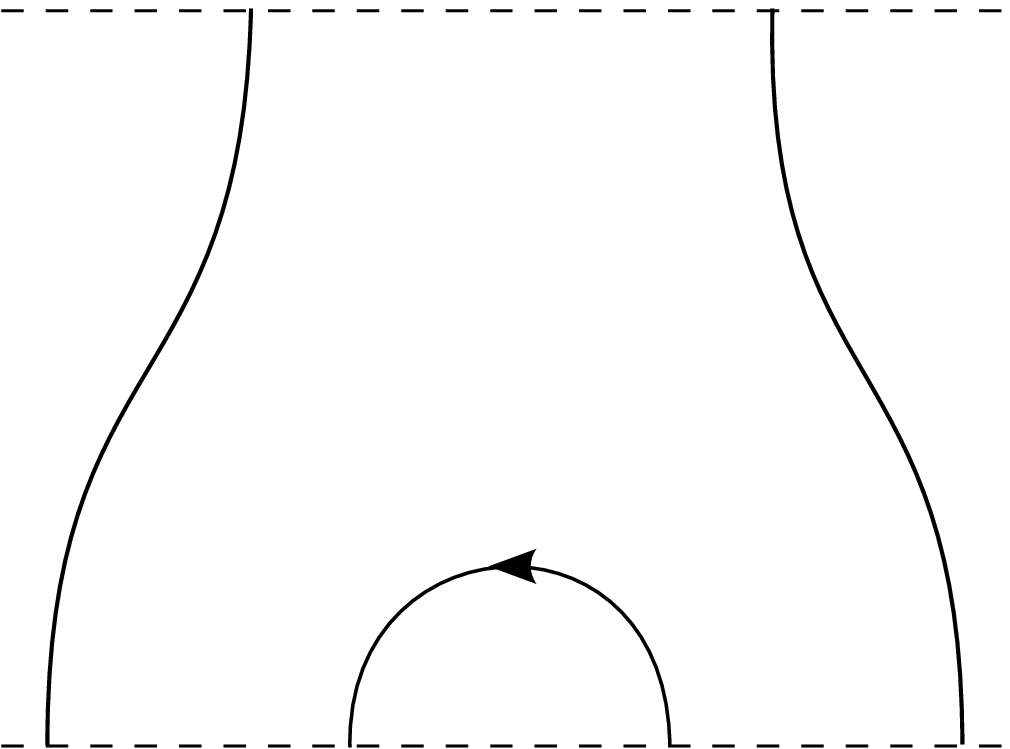}};
  (-21,18)*{m_{A^M_{}} ~:=};
  (0,1.3)*{M};
  (12.7,1.4)*{\bar{M}};
  (25.3,1.3)*{M};
  (37.3,1.3)*{\bar{M}};
  (8.8,36.8)*{M};
  (30.3,36.9)*{\bar{M}};
  (-1,26.4)*{\tc};
  (19,22)*{S};
  (39,26.4)*{\tc};
  (19.5,7.6)*{\tc};
  \endxy }
  \ee
in terms of string diagrams. Similarly, we introduce the coproduct 
$\Delta_{A^M_{}}\colon A^{M}\To{A^{M}}\oti{A^{M}}$ as 
  \be
  \parbox{100pt}{ \xy
  (-21,18)*{\Delta_{A^M_{}} ~:=};
  (20,20)*{\includegraphics[scale=0.4]{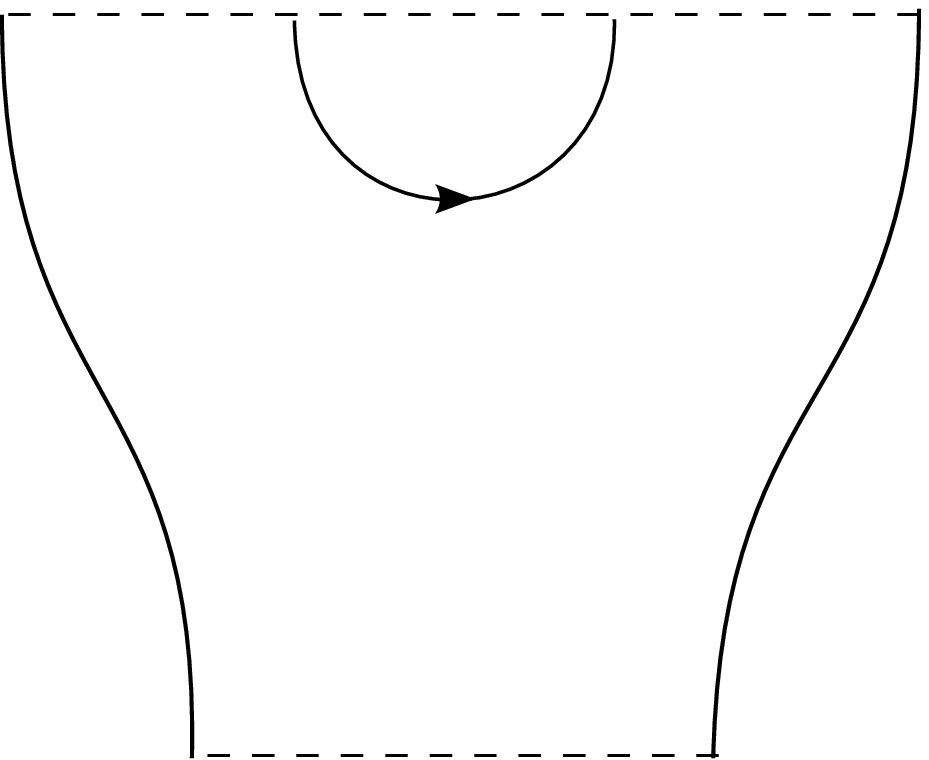}};
  (1.5,37.9)*{M};
  (13.8,38.1)*{\bar{M}};
  (27,37.9)*{M};
  (39.3,38.1)*{\bar{M}};
  (8.7,1.9)*{M};
  (30,2.1)*{\bar{M}};
  (-4,25)*{\tc};
  (20,18)*{S};
  (43,23)*{\tc};
  (20,31.7)*{\tc};
  \endxy } 
  \ee
The morphisms for counit $\eps_{A^M_{}}$ and unit $\eta_{A^M}^{}$ are given by 
  \be
  \parbox{100pt}{ \xy
  (-34,19)*{\eps_{A^M_{}} ~:=};
  (47,19)*{\eta_{A^M_{}} ~:=};
  (0,20)*{\includegraphics[scale=0.4]{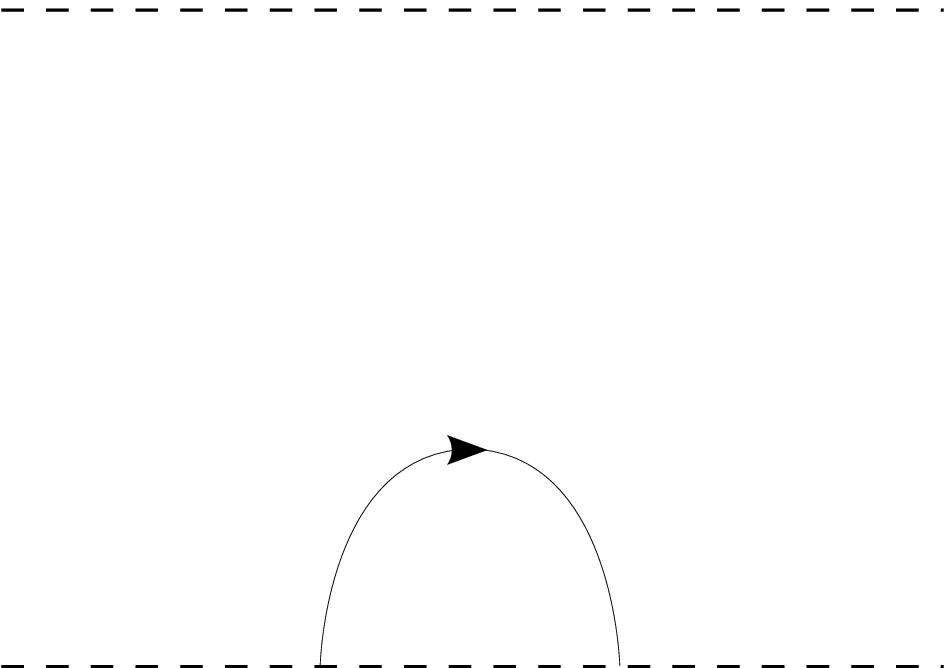}};
  (80,28.5)*{\includegraphics[scale=0.4]{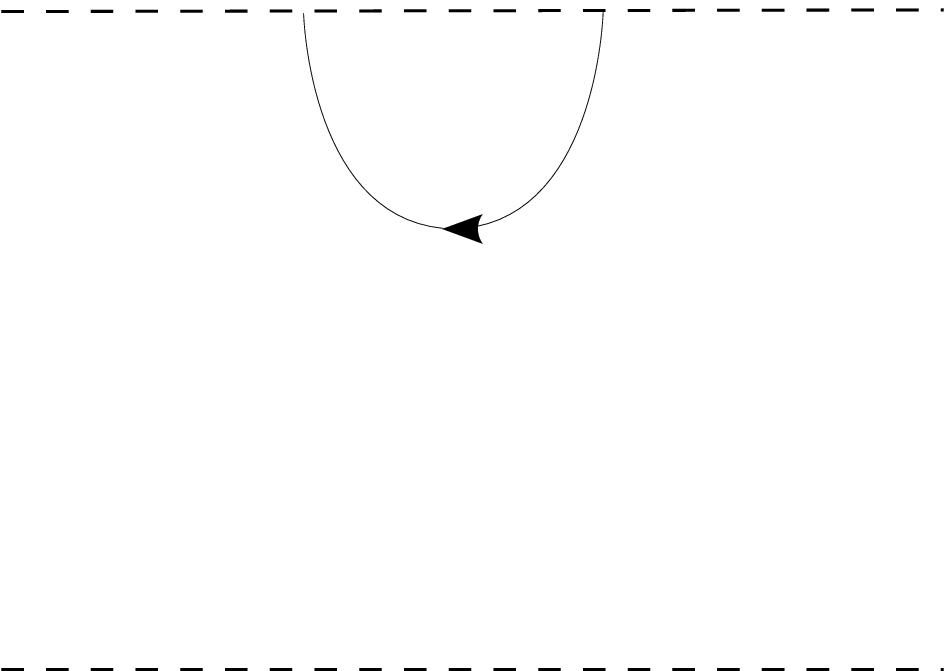}};
  (-6.8,3.8)*{M};
  (5.2,3.8)*{\bar{M}};
  (7.6,25)*{\tc};
  (-.4,10.7)*{S};
  (73.9,36)*{M};
  (86.0,36)*{\bar{M}};
  (95,21)*{\tc};
  (79.9,29.5)*{S};
  \endxy }  
  \ee
It should be appreciated that rescaling the adjunction morphisms 
rescales product and coproduct, and unit and counit, by inverse factors.

\begin{prop}
Let $S$ be an object in \C-\MOD\ corresponding to an exact module category. 
Then for any $M\iN\HOM(S,\cc)$ the morphisms 
$m_{A^M_{}}$, $\eta_{A^M_{}}$, $\Delta_{A^M_{}}$ and $\eps_{A^M_{}}$ 
just introduced endow
the object $A^M$ with the structure of a symmetric Frobenius algebra in \C.
\end{prop}

\begin{proof}
The equality
  \be
  \parbox{100pt}{ \xy
  (-16.5,20)*{\includegraphics[scale=0.7]{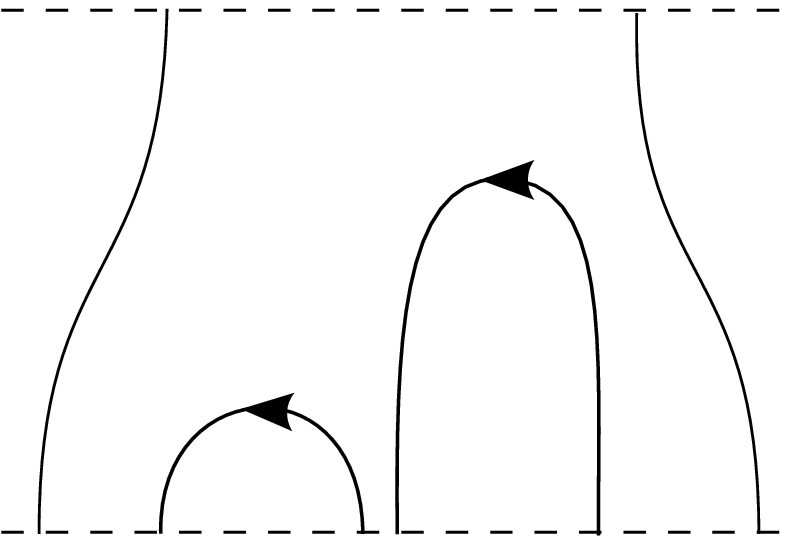}};
  (56.5,20)*{\includegraphics[scale=0.7]{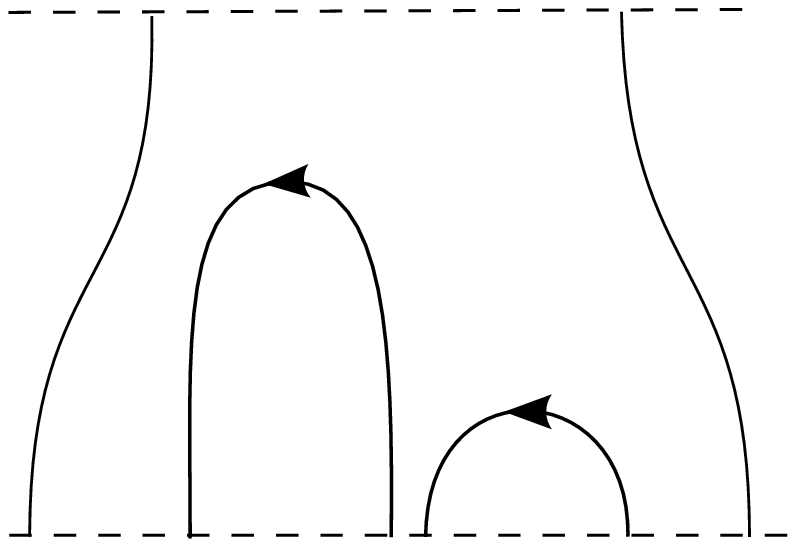}};
  (20,20)*{=};
  (-42,-1.8)*{M};
  (-34,-1.6)*{\bar{M}};
  (-19.6,-1.8)*{M};
  (-15.0,-1.6)*{\bar{M}};
  (-2.3,-1.8)*{M};
  (10,-1.6)*{\bar{M}};
  (30,-1.9)*{M};
  (41.6,-1.6)*{\bar{M}};
  (55.4,-1.8)*{M};
  (59.7,-1.6)*{\bar{M}};
  (72.4,-1.8)*{M};
  (81.7,-1.6)*{\bar{M}};
  (-32.3,41.6)*{M};
  (1.1,41.8)*{\bar{M}};
  (39.6,41.6)*{M};
  (73.2,41.8)*{\bar{M}};
  (-42,30)*{\tc};
  (-24,27)*{S};
  (-26,5.6)*{\tc};
  (-8.7,13)*{\tc};
  (9,30)*{\tc};
  (32,30)*{\tc};
  (64,27)*{S};
  (48.9,13)*{\tc};
  (66.2,5.6)*{\tc};
  (81,30)*{\tc};
  \endxy }
  \ee
which follows from the properties of string diagrams, shows that
the product is associative. Coassociativity of $\Delta_{A^M_{}}$ 
is seen in an analogous manner. The equalities 
  \be
  \parbox{200pt}{ \xy
  (-31.5,20)*{\includegraphics[scale=0.8]{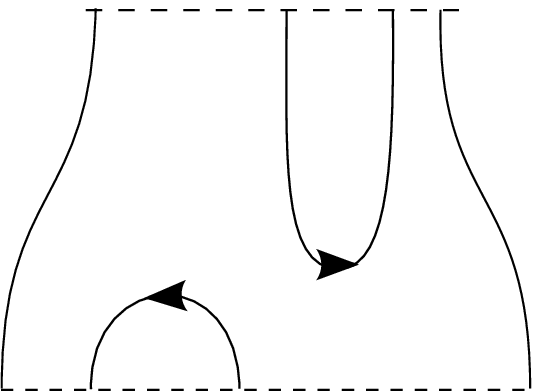}};
  (21.3,20)*{\includegraphics[scale=0.8]{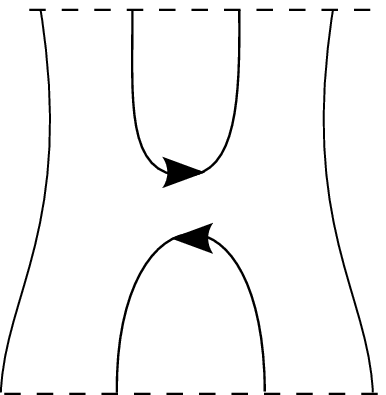}};
  (71.5,20)*{\includegraphics[scale=0.8]{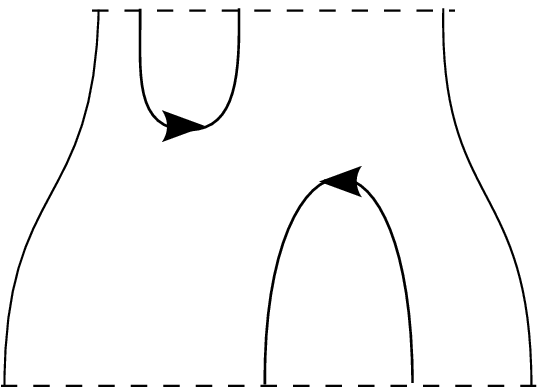}};
  (-45,38.4)*{M};
  (-29.8,38.6)*{\bar{M}};
  (-20.9,38.4)*{M};
  (-16,38.6)*{\bar{M}};
  (-2,20)*{=};
  (10,38.4)*{M};
  (17,38.6)*{\bar{M}};
  (26,38.4)*{M};
  (33.6,38.6)*{\bar{M}};
  (43,20)*{=};
  (57.3,38.4)*{M};
  (62,38.6)*{\bar{M}};
  (69.2,38.4)*{M};
  (86.2,38.6)*{\bar{M}};
  (-53.78,1.4)*{M};
  (-46,1.6)*{\bar{M}};
  (-32.3,1.4)*{M};
  (-10.4,1.6)*{\bar{M}};
  (5.5,1.4)*{M};
  (15,1.6)*{\bar{M}};
  (27.1,1.4)*{M};
  (36,1.6)*{\bar{M}};
  (50,1.4)*{M};
  (71,1.6)*{\bar{M}};
  (82.7,1.4)*{M};
  (92.6,1.6)*{\bar{M}};
  (-51,27)*{\tc};
  (-39.7,7.8)*{\tc};
  (-38.8,24)*{S};
  (-25.7,27)*{\tc};
  (-12.2,28)*{\tc};
  (5.5,27)*{\tc};
  (14.1,19)*{S};
  (20.9,29)*{\tc};
  (21.3,10.2)*{\tc};
  (36.6,28)*{\tc};
  (51.7,27)*{\tc};
  (61.1,15)*{S};
  (65.1,31)*{\tc};
  (77.0,12.9)*{\tc};
  (90.8,28)*{\tc};
  \endxy }
  \ee
prove the Frobenius property.
 \\
Finally, \C\ is rigid, and left and right duals coincide. It is not difficult 
to see that by construction the algebra $A^{M}$ is equal to its dual. The compositions
  \be\label{duality}
  \eps_{A^M} \circ m_{A^M}:\quad A^M \oti A^M \to \one \qquad{\rm and}\qquad
  \Delta_{A^M} \circ \eta_{A^M}:\quad \one \to A^M \oti A^M 
  \ee
give the duality morphisms.  
\end{proof}

There exists a particularly interesting subclass of surface defects
for which the Frobenius algebra $A^M$ obtained from any Wilson line
has additional properties. We need first the

\begin{defi}
A surface defect $S$ in \C-\MOD\ is said to be \emph{special} iff
  \be
  2\mbox-\Hom(\Id_S,\Id_S) \,\simeq\, \complex \,.
  \ee
\end{defi} 

The transparent Wilson line inside a special surface defect $S$ can
only have multiples of the identity as insertions. 
Put differently, there are no non-trivial
local excitations on a surface defect of type $S$ other than those
related to Wilson lines and their junctions.

As an application of this definition, we consider the following situation
in a special surface defect: there is a hole punched out, i.e.\
the surface contains a disk labeled by the transparent defect $\tc$; 
the label for the boundary of the disk is a Wilson line $M$. 
Since there are no local excitations,
we can replace the punched-out hole by the surface defect $S$, provided 
that we multiply every expression obtained with this replacement
by a scalar factor depending on the Wilson line $M$. For the moment we 
cannot yet tell whether this scalar factor is non-zero.
 
Wilson lines separating special defects from the transparent defect 
should yield Frobenius algebras with a particular property. Recall 
from section 2 that a \emph{special algebra} $A$ in a monoidal category 
\C\ is an object which is both an algebra and a coalgebra and satisfies
  \be
  \eps \circ \eta = \beta_{1}\:\id_{\one} \qquad\textrm{and}\qquad
   m\circ\Delta = \beta_{A}\:\id_{A}
  \ee
with non-zero complex numbers $\beta_{1}$ and $\beta_{A}$.

\begin{prop}
Let $S$ be a special surface defect described by a semisimple
module category $S$ over a modular tensor category \,\C.
Then for any Wilson line $M\iN\HOM(S,\cc)$ the corresponding symmetric 
Frobenius algebra $A^{M}$ in \C\ is a special algebra.
\end{prop}

\begin{proof}
For the algebra $A^M$, the composition $m\,{\circ}\,\Delta$ of product and 
coproduct is described by
a string diagram with a hole in the surface defect $S$
whose boundary is labeled by $M$. Since $S$ is special, there are no
local excitations and thus the diagram can be replaced, up to
a scalar factor $\beta_{A^M}$, by a diagram without hole.
On the other hand, since the tensor unit of the modular tensor category 
\C\ is simple, the composition $\eps\,{\circ}\,\eta$ of counit and unit
of $A^M$ is a multiple $\beta_1$ of the identity morphism $\id_\one$.
 \\
Both $\beta_{1}$ and $\beta_{A^M}$ depend on the choices of adjunction 
and coadjunction 2-morphisms for $\bar{M}$. However, computing the 
quantum dimension of $A^M$ using the duality morphisms (\ref{duality}), 
we obtain
  \be
  \dim(A^{M})=\beta_{1}\:\beta_{A^{M}} \,,
  \ee
so that the product of the two scalars is independent of the choices
of adjunction 2-morphisms. Since the object $A^M$ has been constructed 
as a composition of a non-vanishing module functor and its adjoint, $A^M$ 
is not the zero object. As the only object of a modular tensor category 
having vanishing quantum dimension is the zero object, we conclude 
that both scalars $\beta_{1}$ and $\beta_{A^{M}}$ are non-zero.
Hence the algebra $A^M$ is special.  
\end{proof}

We next investigate how the Frobenius algebras $A^{M}$ for a fixed surface
defect $S$ depends on the choice of Wilson line $M$.

\begin{prop}
Let $S$ be a surface defect in \C-\MOD\, and let $M,M'\iN\HOMC(S,\cc)$
be Wilson lines separating $S$ from the transparent defect $\tc$.
 Then the symmetric Frobenius algebras $A^{M}$ and $A^{M'}$ 
are Morita equivalent. 
\end{prop}

\begin{proof}
We explicitly construct a Morita context. Consider the objects 
  \be
  B := M\circ\bar{M'} \qquad\text{and}\qquad 
  \tilde{B} := M'\circ\bar{M} 
  \ee
in $\ENDC(\cc) \,{\simeq}\, \C$.
The counit of the adjunction for $M$ provides a morphism 
  \be
  M\circ \bar M  \circ M \circ \bar M' \to M\circ \bar M' \,,
  \qquad{\rm i.e.}\quad
  A^M\otimes B \to B
  \,.
  \ee
With the help of the isotopy invariance of string diagrams, one quickly 
checks that this morphism obeys the axiom for a left action of $A^M$ on
$B$. This type of argument can be repeated to show that $B$ has
a natural structure of an $A^{M}$-$A^{M'}$-bimodule, and that
$\tilde{B}$ has the structure of an $A^{M'}$-$A^{M}$-bimodule. 
 \\
We next must procure an isomorphism $B\,{\otimes_{\!A^{M'}}}\, \tilde B 
\To A^M$ of bimodules. This is achieved by showing that the morphism
  \be \label{eq6.25}
  M\circ \bar M' \circ M' \circ \bar M \to M\circ\bar M \,, 
  \qquad{\rm i.e.}\quad
  B\otimes \tilde B \to A^M 
  \ee
that is provided by the counit of the adjunction (which is obviously a morphism 
of bimodules) has the universal property of a cokernel. To this end we select
any morphism $\varphi\colon B\oti A^{M'}\oti \tilde B \to X$, with $X$ any 
object of \C, such that 
  \be
  \label{left-right}
  \parbox{100pt}{ \xy
  (5,20)*{\includegraphics[scale=0.6]{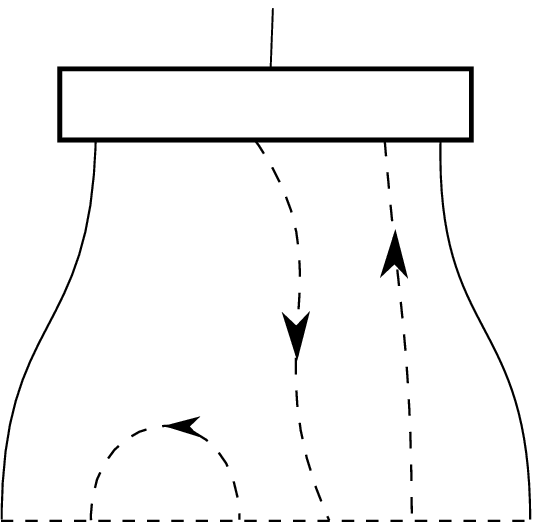}};
  (-14.5,15)*{\tc};
  (-11.7,2.0)*{\scriptstyle M};
  (-5.7,2.0)*{\scriptstyle \bar M'};
  (2.7,2.0)*{\scriptstyle M'};
  (8.6,2.0)*{\scriptstyle \bar M'};
  (14.1,2.0)*{\scriptstyle M'};
  (20.7,2.0)*{\scriptstyle \bar M};
  (-.9,6.8)*{\scriptstyle\tc};
  (.2,18.7)*{S};
  (10.5,12.3)*{\tc};
  (5,29.7)*{\varphi};
  (5.9,38.5)*{X};
  (17.1,9.5)*{S};
  (23.7,16)*{\tc};
  (34,19)*{=};
  (65,20)*{\includegraphics[scale=0.6]{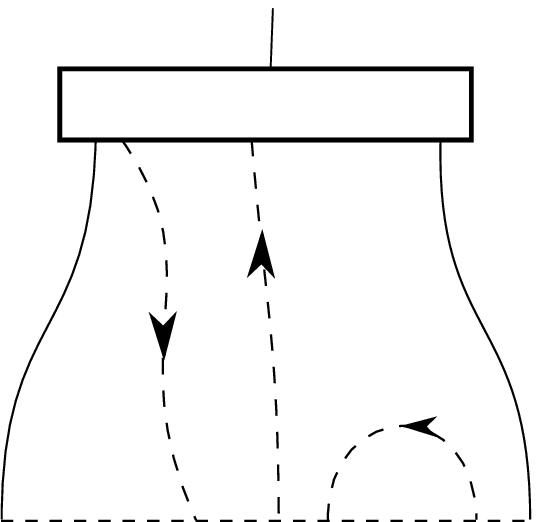}};
  (48.5,2.0)*{\scriptstyle M};
  (60.8,2.0)*{\scriptstyle \bar M'};
  (65.6,2.2)*{\scriptstyle M'};
  (69.4,2.0)*{\scriptstyle \bar M'};
  (77.6,2.0)*{\scriptstyle M'};
  (81.4,2.0)*{\scriptstyle \bar M};
  (45.8,15)*{\tc};
  (54.2,10.4)*{S};
  (62.2,12)*{\tc};
  (65,29.7)*{\varphi};
  (65.9,38.5)*{X};
  (71.2,18.7)*{S};
  (73.7,6.8)*{\scriptstyle\tc};
  (84.3,16)*{\tc};
  \endxy }
  \ee
We are looking for a morphism $\tilde \varphi \colon B \To X $ such that 
$\tilde \varphi \circ (\id_M \oti \eps_{M'} \oti \id_{\bar M}) \eq \varphi$.
Composing this equality with the morphism 
$\id_M \oti \eta_{M'} \oti \id_{\bar M}$ yields
  \be
  \parbox{100pt}{ \xy
  (5,20)*{\includegraphics[scale=0.6]{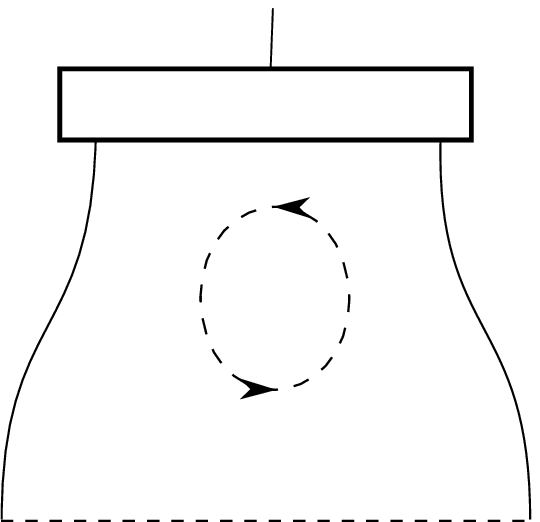}};
  (-11.6,2.0)*{\scriptstyle M};
  (20.7,2.0)*{\scriptstyle \bar M};
  (-14.5,15)*{\tc};
  (5,29.5)*{\tilde\varphi};
  (5.9,38.7)*{X};
  (5.9,18.1)*{\tc};
  (13.3,10.1)*{S};
  (23.7,16)*{\tc};
  (35,19)*{=};
  (65,20)*{\includegraphics[scale=0.6]{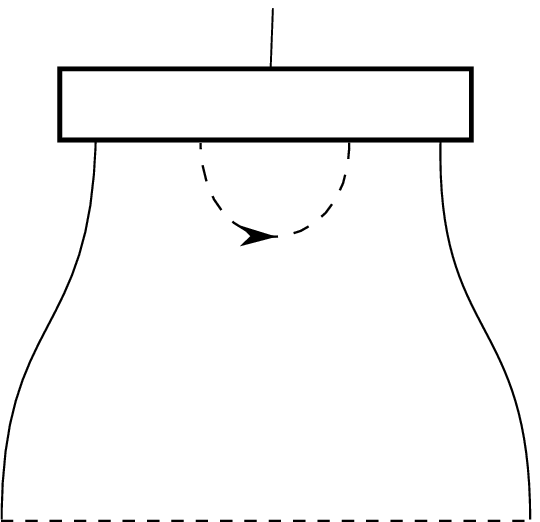}};
  (48.7,2.0)*{\scriptstyle M};
  (80.7,2.0)*{\scriptstyle \bar M};
  (45.8,15)*{\tc};
  (65,29.7)*{\varphi};
  (65.9,38.7)*{X};
  (65.7,24.7)*{\scriptstyle\tc};
  (68.7,13.7)*{S};
  (84.3,14.9)*{\tc};
  \endxy }
  \ee
The left hand side of this equality equals $\beta_{A^{M'}}\,\tilde\varphi$. 
This shows that the morphism $\tilde \varphi$ is uniquely determined.
To establish that $A^M$ is indeed a cokernel, we have to show that the morphism
$\beta_{A^{M'}}^{-1}\, \varphi \circ [\,\id_M \oti (\eta_{M'}\,{\circ}\,\eps_{M'})
\oti \id_{\bar M}\,]$, which is the composition of $\tilde\varphi$ with the
cokernel morphism, equals $\varphi$. This is established by
  \be
  \parbox{100pt}{ \xy
  (5,20)*{\includegraphics[scale=0.6]{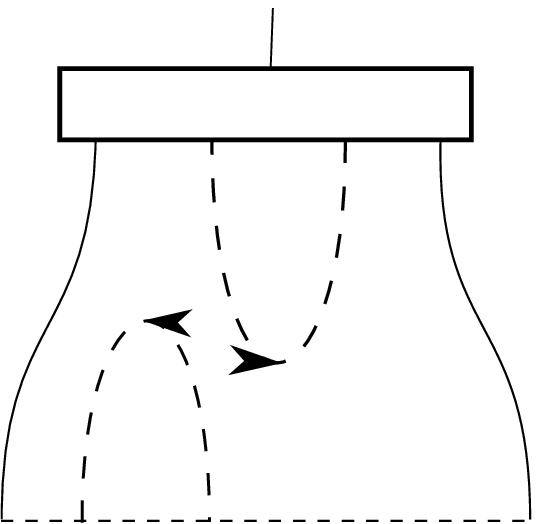}};
  (-11.6,2.0)*{\scriptstyle M};
  (-6.3,2.0)*{\scriptstyle \bar M'};
  (1.5,2.0)*{\scriptstyle M'};
  (20.7,2.0)*{\scriptstyle \bar M};
  (-14.5,15)*{\tc};
  (-2.4,9.7)*{\tc};
  (5,29.7)*{\varphi};
  (5.9,38.7)*{X};
  (5.7,22.6)*{\tc};
  (12.2,11.2)*{S};
  (23.7,16)*{\tc};
  (35,19)*{=};
  (65,20)*{\includegraphics[scale=0.6]{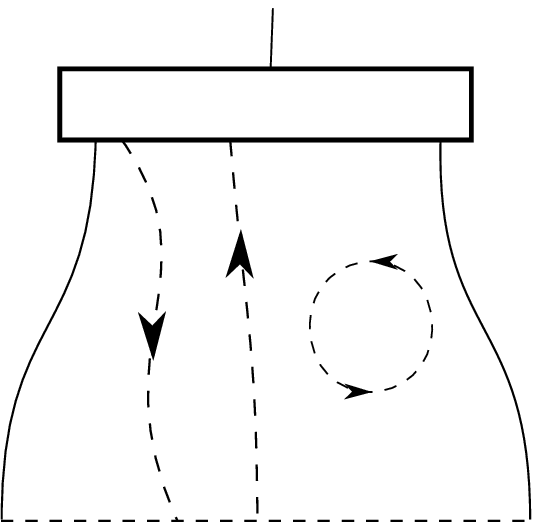}};
  (48.9,2.0)*{\scriptstyle M};
  (59.3,2.0)*{\scriptstyle \bar M'};
  (64.9,2.0)*{\scriptstyle M'};
  (80.7,2.0)*{\scriptstyle \bar M};
  (45.8,15)*{\tc};
  (53.7,8.7)*{S};
  (61.2,11.1)*{\tc};
  (65,29.7)*{\varphi};
  (65.9,38.7)*{X};
  (71.6,16.0)*{\scriptstyle\tc};
  (75.3,8.7)*{S};
  (84.3,14.9)*{\tc};
  (104,19)*{=~~ \beta_{A^{M'}} \varphi~.};
  \endxy }
  \ee
Here in the first equality a right action of $A^{M'}$ composed with
$\varphi$ is replaced by a left action, as in \eqref{left-right}. The 
second equality uses the fact that the defect $S$ is special, so as 
to remove the bubble at the expense of a factor of $\beta_{A^{M'}}$. 
A similar argument shows that
$\tilde B\,{\otimes_{A^M}}\, B \cong A^{M'}$. This completes the proof.
\end{proof}

We can summarize the findings of this section in the 

\begin{thm}
Consider the three-dimensional topological field theory corresponding to
a modular tensor category. To any surface defect separating the TFT from 
itself there is associated a Morita equivalence class of special 
symmetric Frobenius algebras. 
\end{thm}

\vfill

{\small
\noindent{\sc Acknowledgments:}
We are grateful to Alexei Davydov, Juan Mart\'\i n Mombelli, Thomas Nikolaus,
Viktor Ostrik, Ingo Runkel and Gregor Schaumann
for helpful discussions. JF is largely supported 
by VR under project no.\ 621-2009-3993. CS and AV are partially supported 
by the Collaborative Research Centre 676 ``Particles, Strings and the Early 
Universe - the Structure of Matter and Space-Time'' and by the DFG Priority 
Programme 1388 ``Representation Theory''. JF is grateful to Hamburg 
University, and in particular to CS, Astrid D\"orh\"ofer and Eva Kuhlmann, 
for their hospitality when part of this work was done.
\\
We particularly acknowledge stimulating discussions and talks at the
``Workshop on Representation Theoretical and Categorical Structures in 
Quantum Geometry and Conformal Field Theory'' at the University of 
Erlangen in November 2011, which was supported by the DFG Priority 
Programme 1388 ``Representation 
Theory'' and the ESF network ``Quantum Geometry and Quantum Gravity''.
}

\newpage

 \newcommand\wb{\,\linebreak[0]} \def\wB {$\,$\wb}
 \newcommand\Bi[2]    {\bibitem[#2]{#1}} 
 \newcommand\inBo[8]  {{\em #8}, in:\ {\em #1}, {#2}\ ({#3}, {#4} {#5}), p.\ {#6--#7} }
 \newcommand\inBO[9]  {{\em #9}, in:\ {\em #1}, {#2}\ ({#3}, {#4} {#5}), p.\ {#6--#7} {\tt [#8]}}
 \newcommand\J[7]     {{\em #7}, {#1} {#2} ({#3}) {#4--#5} {{\tt [#6]}}}
 \newcommand\JO[6]    {{\em #6}, {#1} {#2} ({#3}) {#4--#5} }
 \newcommand\JP[7]    {{\em #7}, {#1} ({#3}) {{\tt [#6]}}}
 \newcommand\BOOK[4]  {{\em #1\/} ({#2}, {#3} {#4})}
 \newcommand\PhD[2]   {{\em #2}, Ph.D.\ thesis #1}
 \newcommand\Prep[2]  {{\em #2}, preprint {\tt #1}}
 \def\adma  {Adv.\wb Math.}
 \def\amjm  {Amer.\wb J.\wb Math.}   
 \def\anma  {Ann.\wb Math.}
 \def\aspm  {Adv.\wb Stu\-dies\wB in\wB Pure\wB Math.}
 \def\coma  {Con\-temp.\wb Math.}
 \def\comp  {Com\-mun.\wb Math.\wb Phys.}
 \def\ijmp  {Int.\wb J.\wb Mod.\wb Phys.\ A}
 \def\imrn  {Int.\wb Math.\wb Res.\wb Notices}
 \def\inma  {Invent.\wb math.}
 \def\jajm  {Japan.\wb J.\wb Math.}
 \def\jgap  {J.\wb Geom.\wB and\wB Phys.}
 \def\jhep  {J.\wb High\wB Energy\wB Phys.}
 \def\joal  {J.\wB Al\-ge\-bra}
 \def\jopa  {J.\wb Phys.\ A}
 \def\jktr  {J.\wB Knot\wB Theory\wB and\wB its\wB Ramif.}
 \def\jpaa  {J.\wB Pure\wB Appl.\wb Alg.}
 \def\jram  {J.\wB rei\-ne\wB an\-gew.\wb Math.}
 \def\leni  {Lenin\-grad\wB Math.\wb J.}
 \def\nupb  {Nucl.\wb Phys.\ B}
 \def\pams  {Proc.\wb Amer.\wb Math.\wb Soc.}
 \def\phlb  {Phys.\wb Lett.\ B}
 \def\pspm  {Proc.\wb Symp.\wB Pure\wB Math.}
 \def\quto  {Quantum\wB Topology}
 \def\ruma  {Revista de la Uni\'on Matem\'atica Argentina}
 \def\sema  {Selecta\wB Mathematica}
 \def\slnm  {Sprin\-ger\wB Lecture\wB Notes\wB in\wB Mathematics}
 \def\taac  {Theo\-ry\wB and\wB Appl.\wb Cat.}
 \def\trgr  {Trans\-form.\wB Groups}

\end{document}